\documentclass{bmcart}

\usepackage{amssymb,xspace,amsmath,amsthm}
\usepackage[utf8]{inputenc}
\usepackage{bm}
\usepackage{url} 
\usepackage{color}
\usepackage{multirow}
\usepackage{mathpazo}

\usepackage{lscape}

\usepackage[pdftex]{graphicx}
\usepackage[final]{pdfpages}
\usepackage{graphicx}

\startlocaldefs
\newtheorem{definition}{Definition}

\newtheorem{lemma}{Lemma}

\newtheorem{theorem}{Theorem}

\endlocaldefs

\endlocaldefs

\begin{document}

%%% Start of article front matter
\begin{frontmatter}

\begin{fmbox}
\dochead{Research}


\title{Aligning coding sequences with frameshift extension penalties}

%%%%%%%%%%%%%%%%%%%%%%%%%%%%%%%%%%%%%%%%%%%%%%
%%%%
%% Enter the authors here %%
%%%%
%% Specify information, if available,       %%
%% in the form:  %%
%%   <key>={<id1>,<id2>}  %%
%%   <key>=%%
%% Comment or delete the keys which are     %%
%% not used. Repeat \author command as much %%
%% as required.  %%
%%%      %%
%%%%%%%%%%%%%%%%%%%%%%%%%%%%%%%%%%%%%%%%%%%%%%

\author[
   addressref={aff1},% id's of addresses, e.g. {aff1,aff2}
   corref={aff1},     % id of corresponding address, if any
   %noteref=%{ },      % id's of article notes, if any
   email={safa.jammali@USherbrooke.ca}  % email address
]{\inits{SJ}\fnm{Safa} \snm{Jammali}}
\author[
   addressref={aff1},
% id's of addresses, e.g. {aff1,aff2}
   email={esaie.kuitche.kamela@USherbrooke.ca}   % email address
]{\inits{FB}\fnm{Esaie} \snm{Kuitche}}
\author[
   addressref={aff1}, % id's of addresses, e.g. {aff1,aff2}
   email={ayoub.rachati@USherbrooke.ca}   % email address
]{\inits{AR} \fnm{Ayoub} \snm{Rachati}}
\author[
   addressref={aff1},
  % id's of addresses, e.g. {aff1,aff2}
   email={francois.belanger3@USherbrooke.ca}   % email address
]{\inits{FB}\fnm{Fran\c{c}ois} \snm{B\'elanger}}
\author[
   addressref={aff2}, % id's of addresses, e.g. {aff1,aff2}
   email={michelle.scott@USherbrooke.ca}   % email address
]{\inits{MS}\fnm{Michelle} \snm{Scott}} 
\author[
   addressref={aff1}, % id's of addresses, e.g. {aff1,aff2}
   email={aida.ouangraoua@USherbrooke.ca}   % email address
]{\inits{AO}\fnm{ A\"ida} \snm{Ouangraoua}} 

%%%%%%%%%%%%%%%%%%%%%%%%%%%%%%%%%%%%%%%%%%%%%%
%%%%
%% Enter the authors' addresses here        %%
%%%%
%% Repeat \address commands as much as      %%
%% required.     %%
%%%%
%%%%%%%%%%%%%%%%%%%%%%%%%%%%%%%%%%%%%%%%%%%%%%

\address[id=aff1] {%        % unique id
  \orgname{D\'epartement d'informatique, Facult\'e des Sciences,  Universit\'e de Sherbrooke, Sherbrooke},  
  \postcode{J1K2R1}    
  \city{QC},  
  \cny{Canada}  
}
\address[id=aff2] {%        % unique id
  \orgname{D\'epartement de biochimie, Facult\'e  de m\'edecine et des sciences de la sant\'e, Universit\'e de Sherbrooke, Sherbrooke},  
  \postcode{J1E4K8}    
  \city{QC},  
  \cny{Canada}  
}

%%%%%%%%%%%%%%%%%%%%%%%%%%%%%%%%%%%%%%%%%%%%%%
%%%%
%% Enter short notes here %%
%%%%
%% Short notes will be after addresses      %%
%% on first page.%%
%%%%
%%%%%%%%%%%%%%%%%%%%%%%%%%%%%%%%%%%%%%%%%%%%%%

\begin{artnotes}
%\note{Sample of title note}     % note to the article
%\note[id=]{} % note, connected to author
\end{artnotes}

\end{fmbox} %comment this for two column layout

%%%%%%%%%%%%%%%%%%%%%%%%%%%%%%%%%%%%%%%%%%%%%%
%%%%
%% The Abstract begins here        %%
%%%%
%% Please refer to the Instructions for     %%
%% authors on http://www.biomedcentral.com  %%
%% and include the section headings%%
%% accordingly for your article type.       %%
%%%%
%%%%%%%%%%%%%%%%%%%%%%%%%%%%%%%%%%%%%%%%%%%%%%

\begin{abstractbox}

\begin{abstract} % abstract
\parttitle{Background} 
  Frameshift translation is an important phenomenon that contributes to the appearance
  of novel Coding DNA Sequences (CDS) and functions in gene evolution, by allowing 
  alternative amino acid translations of gene coding regions. 
  
  Frameshift translations can be identified by aligning two CDS, from a same gene or from homologous genes, while accounting for their 
  codon structure. Two main classes of algorithms have been proposed to solve the problem of 
  aligning CDS, either by amino acid sequence alignment back-translation, or by simultaneously accounting 
  for the nucleotide and amino acid levels. The former does not allow to 
  account for frameshift translations and up to now, the latter exclusively accounts for frameshift
  translation initiation, not considering the length of the translation disruption 
  caused by a frameshift.
\parttitle{Results} 
  We introduce a new scoring scheme with an algorithm for the 
  pairwise alignment of CDS accounting for frameshift translation initiation and length, while simultaneously considering nucleotide and amino acid sequences. The main specificity 
  of the scoring scheme is the introduction of a penalty cost accounting for frameshift 
  extension length to compute an adequate similarity score for a CDS alignment.  The second 
  specificity of the model is that the search space of the problem solved is the set of all 
  feasible alignments between two CDS. Previous approaches have considered restricted search 
  space or additional constraints on the decomposition of an alignment into length-3
  sub-alignments. The algorithm described in this paper has the same
  asymptotic time complexity as the classical Needleman-Wunsch algorithm.
  
  \parttitle{Conclusions} 
  We compare the method to other CDS alignment methods based on an application to 
  the comparison of pairs of CDS from homologous \emph{human}, \emph{mouse} 
  and \emph{cow} genes of ten mammalian gene families from the Ensembl-Compara database.   
   The results show that our method is particularly robust to parameter changes as compared to existing methods. It also appears to be a good compromise, performing well
both in the presence and absence of frameshift translations.
  An implementation of the method is available at https://github.com/UdeS-CoBIUS/FsePSA.

\end{abstract}

%%%%%%%%%%%%%%%%%%%%%%%%%%%%%%%%%%%%%%%%%%%%%%
%%%%
%% The keywords begin here%%
%%%%
%% Put each keyword in separate \kwd{}.     %%
%%%%
%%%%%%%%%%%%%%%%%%%%%%%%%%%%%%%%%%%%%%%%%%%%%%


% MSC classifications codes, if any
%\begin{keyword}[class=AMS]
%\kwd[Primary ]{}
%\kwd{}
%\kwd[; secondary ]{}
%\end{keyword}

\end{abstractbox}
%
%\end{fmbox}% uncomment this for twcolumn layout

\end{frontmatter}

%%%%%%%%%%%%%%%%%%%%%%%%%%%%%%%%%%%%%%%%%%%%%%
%%%%
%% The Main Body begins here       %%
%%%%
%% Please refer to the instructions for     %%
%% authors on:   %%
%% http://www.biomedcentral.com/info/authors%%
%% and include the section headings%%
%% accordingly for your article type.       %%
%%%%
%% See the Results and Discussion section   %%
%% for details on how to create sub-sections%%
%%%%
%% use \cite{...} to cite references        %%
%%  \cite{koon} and       %%
%%  \cite{oreg,khar,zvai,xjon,schn,pond}    %%
%%  \nocite{smith,marg,hunn,advi,koha,mouse}%%
%%%%
%%%%%%%%%%%%%%%%%%%%%%%%%%%%%%%%%%%%%%%%%%%%%%

%%%%%%%%%%%%%%%%%%%%%%%%% start of article main body
% <put your article body there>

%%%%%%%%%%%%%%%%
%% Background %%
%%
%\section*{Content}
%Text and results for this section, as per the individual journal's instructions for authors. %\cite{koon,oreg,khar,zvai,xjon,schn,pond,smith,marg,hunn,advi,koha,mouse}

\section*{Background}
\label{intro}

Biological sequence alignment is a cornerstone of bioinformatics and is widely used in such fields as phylogenetic reconstruction, gene finding, genome assembly. The accuracy of the 
sequence alignments and similarity measures are directly related to the accuracy
of subsequent analysis. CDS alignment methods have many important applications for gene tree 
and protein tree reconstruction. In fact, they are useful to cluster homologous 
CDS into groups of orthologous splicing isoforms \cite{zambelli2010,barbosa2012} and combine 
partial trees on orthology groups into a complete protein tree for a gene family 
\cite{christinat2013,kuitche2017}. 
Aligning and 
measuring the similarity between homologous CDS requires to account for \emph{Frameshift (FS) 
translations} that cannot be detected at the amino acid (AA) level, but lead to a high 
similarity at the nucleotide level between functionnaly different sub-sequences. 

FS translation consists in alternative AA translations
of a coding region of DNA using different translation frames \cite{pruitt2009consensus}. 
It is an important phenomenon resulting from different scenarios such as, insertion or deletion 
of a nucleotide sequence whose length is not a multiple of $3$ in a CDS through alternative splicing \cite{okamura2006frequent, Barmak2003} or evolutionary genomic indels \cite{Stoffer1997, IKUO1991}, programmed ribosomal frameshifting \cite{Robin2012},  or sequencing errors \cite{We2013}. Recent studies have reported the role of FS translations in the appearance of novel CDS and functions in gene evolution \cite{okamura2006frequent, raes2005functional}. FS translation has also been found to
be linked to several diseases such as the Crohn's Disease \cite{ogu2001}. 
The computational detection of FS translations requires the alignment of CDS while accounting for their codon structure.
A classical approach for aligning two CDS used in most alignment tools \cite{abascal2010,morgenstern2004}
consists in a three-step method, where the CDS are first translated into
AA sequences using their actual coding frame, then AA sequences are aligned,
and finally the AA alignment is back-translated to a CDS alignment. 
This approach does not account for alternative AA translations between two CDS 
and it leads to incorrect alignment of the coding regions subject to FS translation. 
The opposite problem of aligning protein sequences
while recovering their hypothetical nucleotide CDS sequences and accounting for FS
translation was also studied in several papers
\cite{girdea2010back,moreira2004tip}.

Here, we consider the problem of aligning two CDS while accounting for FS
translation, by simultaneously accounting for their nucleotide and 
AA sequences. The problem has recently regained attention due to the increasing 
evidence for alternative protein production through FS translation by eukaryotic gene families \cite{ranwez2011,Danny2013}. 

The problem was first addressed by Hein et al. \cite{hein1994,pedersen1998}
who proposed a DNA/Protein model
such that the score of an alignment between two CDS of length $n$ and $m$
is a combination of its
score at the nucleotide level and its score at the AA level. They described
a $O(n^2m^2)$ algorithm in \cite{hein1994}, later improved to a $O(nm)$
algorithm in \cite{pedersen1998} for computing an optimal score alignment,
under the constraint that the search space of the problem is restricted.
Arvestad \cite{arvestad1997} later proposed a CDS alignment scoring model
with a $O(nm)$ alignment algorithm accounting for codon structures and FS
translations based on the concept of generalized substitutions introduced in
\cite{sankoff1983}. In this model, the score of a CDS alignment depends on
its decomposition into a concatenation of \emph{codon fragment} alignments,
such that a codon fragment of a CDS is defined as a substring of length $w$,
$0\leq w \leq 5$. This decomposition into codon fragment alignments allows
to define a score of the CDS alignment at the AA level. 
More recently, Ranwez et al. \cite{ranwez2011} proposed a simplification of
the model of Arvestad limiting the maximum length of a codon fragment to $3$.
Under this model, a $O(nm)$ CDS alignment algorithm was described and extended
in the context of multiple sequence alignment \cite{ranwez2011}. In the models
of Arvestad \cite{arvestad1997} and Ranwez et al. \cite{ranwez2011}, several
scores may be computed for the same alignment based on different decompositions
into codon fragment alignments. The corresponding algorithms
for aligning two CDS then consist in computing an optimal score decomposition
of an alignment between the two CDS. 
%To the best of our knowledge, this is the 
%only alignment algorithm that account for FS translation. However, its 
This optimal score exclusively  accounts for FS translation initiations, i.e a FS translation 
in an alignment is penalized by adding a constant FS cost, which only penalizes the initiation of the FS, not accounting for the length of this FS translation. % in the alignment.
However, taking account of FS translation lengths is important in order to increase the precision of CDS alignment scores, as these lengths induce more or less disruptions between the protein sequences.

In this paper, we propose the first alignment algorithm that accounts for both the initiation and the length of FS translations in order to compute the similarity scores of CDS alignments.
The remaining of the paper is organized as follows. %In Section \ref{motivation},
In the "Motivation" section, we illustrate the importance of accounting for
FS translation length when aligning CDS.
%In Section \ref{preliminaries}, 
In the "Preliminaries" section, we give some preliminary definitions and we introduce a
new CDS alignment scoring model with a self-contained definition of the score
of an alignment penalizing both the initiation and the extension of FS
translations. 
%In Section \ref{algorithm}, 
In the "Method" section, a dynamic programming algorithm
for computing an optimal score alignment between two CDS is described.
%The search space of the algorithm consists of all possible alignments between
%the two CDS.
Finally, in the "Results" section,%in Section \ref{experiment}, 
we present and discuss the results of a comparison of our method
with other CDS alignment methods for a pairwise comparison
of CDS from homologous %\emph{human}, \emph{mouse} and \emph{cow} 
genes of ten mammalian gene families. 
%The results show that our method performs well
%both in the presence or absence of FS translations between the CDS to align. 
%

\section*{ Motivation: Importance of accounting for FS translation length}
\label{motivation}
The two main goals of aligning biological sequences are to evaluate the
similarity and to identify similar regions between the sequences, used
thereafter to realize molecular analyses such as evolutionary, functional and 
structural predictions. 
%In practice, the sequence similarity is used
%to exhaustively identify the common features of a set of proteins. 
%New proteins are also assigned to protein families by searching for similar sequences in databases. 
In practice, CDS alignment can be used to 
exhaustively identify the conserved features of a set of proteins. 
Thus, the definition of 
CDS similarity must account for sequence conservation and disruptions 
at both the nucleotide and the protein levels.

\begin{figure}[ht!]
	\centering
 \includegraphics[width=0.8\textwidth]{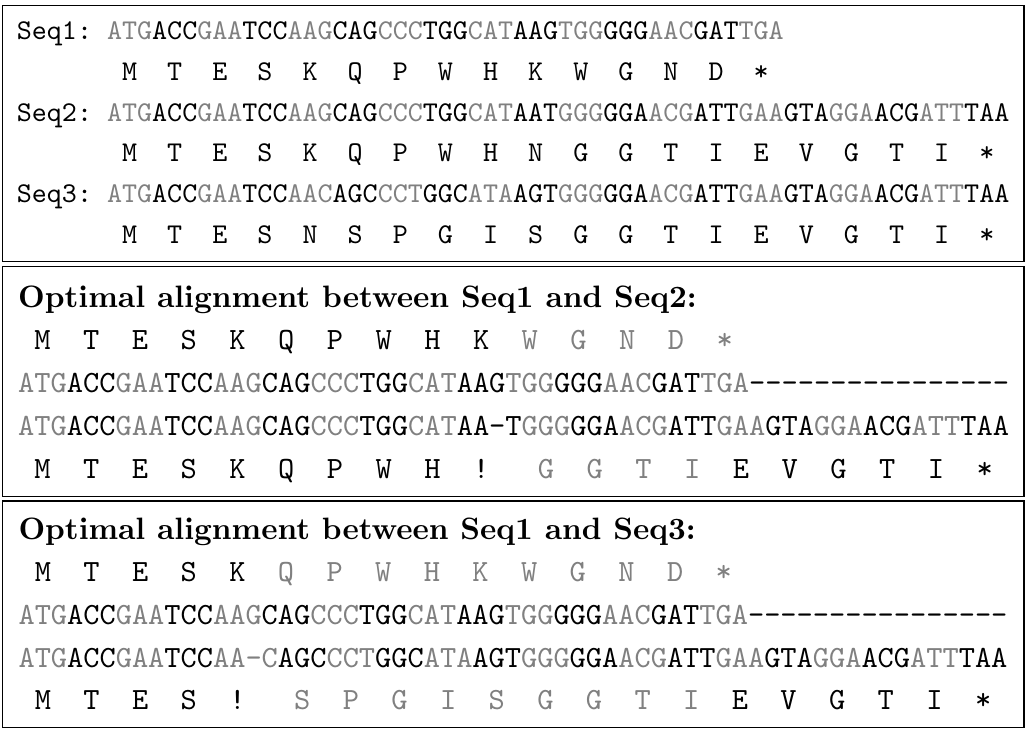}

	\caption{{\bf Top.} An example of three CDS \texttt{Seq1}, \texttt{Seq2} and \texttt{Seq3}. {\bf Middle.} An optimal alignment between \texttt{Seq1} and \texttt{Seq2} with a FS translation region of length $15$. {\bf Bottom.}
    An optimal alignment between \texttt{Seq1} and \texttt{Seq3} with a FS translation region of length $30$.
	}
	\label{fig:importance}
\end{figure}
%A biological function can be attributed to an uncharacterized protein if the latter is in a group where is a protein whose function was previously determined.
Figure \ref{fig:importance} 
illustrates the importance of accounting for AA translations and FS translation length  in order to compute an adequate similarity score for a CDS alignment.
It describes an example of three CDS \texttt{Seq1}, \texttt{Seq2} and \texttt{Seq3}. \texttt{Seq1} has a length of $45$. The CDS \texttt{Seq2}
has length $60$ and is obtained from \texttt{Seq1} by deleting the
nucleotide 'G' at position $30$ and adding $16$ nucleotides at the end.
The CDS \texttt{Seq3} has length
$60$ and is obtained from \texttt{Seq1} by deleting the nucleotide 'G' at position $15$ and adding $16$ nucleotides at the end.

%\noindent
%\texttt{Seq1: \ac{ATG}ACC\ac{GAA}TCC\ac{AAG}CAG\ac{CCC}TGG\ac{CAT}AAG\ac{TGG}GGG\ac{AAC}GAT\ac{TGA}}\\
%\texttt{$~~~~~~~~$M~~T~~E~~S~~K~~Q~~P~~W~~H~~K~~W~~G~~N~~D~~*}\\ 
%\texttt{Seq2: \ac{ATG}ACC\ac{GAA}TCC\ac{AAG}CAG\ac{CCC}TGG\ac{CAT}AAT\ac{GGG}GGA\ac{ACG}ATT\ac{GAA}GTA\ac{GGA}ACG\ac{ATT}TAA}\\
%\texttt{$~~~~~~~~$M~~T~~E~~S~~K~~Q~~P~~W~~H~~N~~G~~G~~T~~I~~E~~V~~G~~T~~I~~*}\\ 
%\texttt{Seq3: \ac{ATG}ACC\ac{GAA}TCC\ac{AAC}AGC\ac{CCT}GGC\ac{ATA}AGT\ac{GGG}GGA\ac{ACG}ATT\ac{GAA}GTA\ac{GGA}ACG\ac{ATT}TAA}\\
%\texttt{$~~~~~~~~$M~~T~~E~~S~~N~~S~~P~~G~~I~~S~~G~~G~~T~~I~~E~~V~~G~~T~~I~~*}\\

When looking at the AA translations of \texttt{Seq1}, \texttt{Seq2} and
\texttt{Seq3}, we observe that the similarity between \texttt{Seq2} and
\texttt{Seq1} is higher than the similarity between \texttt{Seq3} and
\texttt{Seq1} at the protein level, because \texttt{Seq1} and
\texttt{Seq2} share a longer AA prefix \texttt{"M~T~E~S~K~Q~P~W~H"}
(amino acids in black characters in the alignments).
However, the pairwise CDS alignment algorithms that do not account for the length
of FS translations would return the same score for the two following optimal
alignments of \texttt{Seq1} with \texttt{Seq2} and \texttt{Seq1} with
\texttt{Seq3}, penalizing only the initiation of one FS translation in both
cases (positions marked with a "!" symbol in the alignments), and not penalizing the sequence disruptions at the protein level.

%\newpage
%\noindent
%\texttt{\bf Optimal alignment between Seq1 and Seq2:}\\
%\texttt{$~~${M~~T~~E~~S~~K~~Q~~P~~W~~H}~~K~~\ac{W~~G~~N~~D~~*}}\\
%\texttt{\ac{ATG}ACC\ac{GAA}TCC\ac{AAG}CAG\ac{CCC}TGG\ac{CAT}AAG\ac{TGG}GGG\ac{AAC}GAT\ac{TGA}----------------}\\
%\texttt{\ac{ATG}ACC\ac{GAA}TCC\ac{AAG}CAG\ac{CCC}TGG\ac{CAT}AA-T\ac{GGG}GGA\ac{ACG}ATT\ac{GAA}GTA\ac{GGA}ACG\ac{ATT}TAA}\\
%\texttt{$~~${M~~T~~E~~S~~K~~Q~~P~~W~~H}~~!~~\ac{~G~~G~~T~~I}~~E~~V~~G~~T~~I~~*~~}\\
%
%\noindent
%\texttt{\bf Optimal alignment between Seq1 and Seq3:}\\
%\texttt{$~~${M~~T~~E~~S}~~K~~\ac{Q~~P~~W~~H~~K~~W~~G~~N~~D~~*}}\\
%\texttt{\ac{ATG}ACC\ac{GAA}TCC\ac{AAG}CAG\ac{CCC}TGG\ac{CAT}AAG\ac{TGG}GGG\ac{AAC}GAT\ac{TGA}----------------}\\
%\texttt{\ac{ATG}ACC\ac{GAA}TCC\ac{AA-C}AGC\ac{CCT}GGC\ac{ATA}AGT\ac{GGG}GGA\ac{ACG}ATT\ac{GAA}GTA\ac{GGA}ACG\ac{ATT}TAA}\\
%\texttt{$~~${M~~T~~E~~S}~~!~~\ac{~S~~P~~G~~I~~S~~G~~G~~T~~I}~~E~~V~~G~~T~~I~~*} \\
From an evolutionary point of view, a good scoring model for evaluating the similarity 
between two CDS
in the presence of FS translations should then penalize not only the initiation of
FS but also the length of FS translations extension (amino acids in gray
characters in the alignments). The alignment of \texttt{Seq1} with \texttt{Seq2} would
then have a higher similarity score than the alignment of \texttt{Seq1} with
\texttt{Seq3}.
\section*{Preliminaries: Score of CDS alignment}
\label{preliminaries}
In this section, we formally describe a new definition
of the score of a CDS alignment that penalizes both the initiation and the
extension of FS translations.
\begin{definition}[Coding DNA sequence (CDS)]  \hfill\\
\label{def:CDS}
 A coding DNA sequence (CDS) is a DNA sequence on the alphabet of nucleotides
  $\Sigma_N=\{A,C,G,T\}$ whose length $n$ is a multiple of $3$. A coding
  sequence is composed of a concatenation of $\frac{n}{3}$ codons that are
  the words of length $3$ in the sequence ending at positions $3i$,
  $1 \leq i \leq \frac{n}{3}$. The AA translation of a CDS is a protein
  sequence of length $\frac{n}{3}$ on the alphabet $\Sigma_A$ of AA such that
  each codon of the CDS is translated into an AA symbol in the
  protein sequence. 
 \end{definition}
Note that, in practice an entire  CDS begins with a start codon \texttt{"ATG"} 
and ends with a stop codon \texttt{"TAA"}, \texttt{"TAG"} or \texttt{"TGA"}. 
%The following definitions are independent from this feature.

\begin{definition}[Alignment between DNA sequences] \hfill\\
An alignment between two DNA sequences $A$ and $B$ is a pair
$(A',B')$ where $A'$ and $B'$ are two sequences of same length $L$
derived by inserting gap symbols $'-'$ in $A$ and $B$, such that
$\forall i, ~1 \leq i \leq L,  ~ A'[i] \neq ~'-'$ or $B'[i] \neq ~'-'$.
Each position $i, ~1 \leq i \leq L$, in the alignment is called a
column of the alignment.
\end{definition}

Given an alignment $(A',B')$ of length $L$ between two CDS $A$ and $B$,
let $S$ be the sequence $A'$ or $B'$. We denote by
$S[k~..~l], ~1 \leq k \leq l \leq L$, the substring of $S$ going from
position $k$ to position $l$. $|S[k~..~l]|$ denotes the number of letters
in $S[k~..~l]$ that are different from the gap symbol $'-'$.
For example, if $A'=\texttt{ACCAT--GTAG}$ and $B'=\texttt{AC--TACGTAG}$, $|A'[4~..~8]| = |\texttt{AT--G}| = 3$.
A codon of $A$ or $B$ is \emph{grouped in the alignment} $(A',B')$ if its
three nucleotides appear in three consecutive columns of the alignment.
For example, the first codon \texttt{ACC} of $A$  is grouped, while the first codon \texttt{ACT} of $B$ is not grouped.

In the following, we give our definition of the score of an alignment $(A',B')$
between two CDS $A$ and $B$. It is based on a partition of the codons
of $A$ (resp. $B$) into four sets depending on the alignment of codons
(see Figure \ref{fig:CodonSets} for an illustration):

\begin{enumerate}
\item  The set of \texttt{In-frame Matching codons (IM)} contains the codons
that are grouped in the alignment and aligned with a
codon of the other CDS. 
\item The set of \texttt{Frameshift extension codons (FSext)} contains the
  codons
that are grouped in the alignment and aligned with a concatenation of three 
nucleotides that overlaps two codons of the other CDS.
\item The set of \texttt{Deleted/Inserted codons (InDel)} contains the codons
that are grouped in the alignment and aligned with a concatenation of $3$ gap
symbols.
\item All other codons constitutes
  \texttt{Frameshift initiation codons (FSinit)}. The set of
  \texttt{Matching nucleotides} \texttt{in FSinit codons (MFS)} contains
all the nucleotides belonging to \texttt{FSinit} codons and aligned with a
nucleotide of the other CDS. 
\end{enumerate}

\begin{figure}[ht!]
	\centering
	\includegraphics[width=0.8\textwidth]{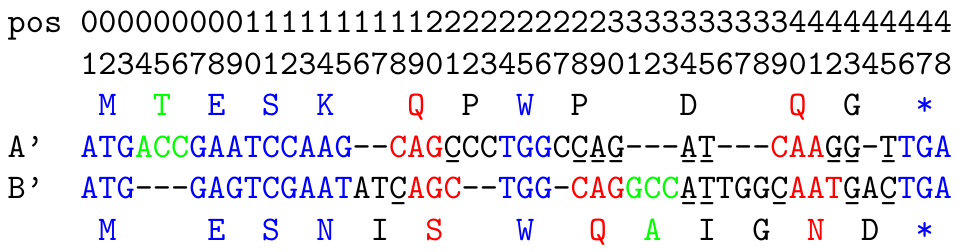}
	\caption{An alignment $(A',B')$ of length $48$ between two CDS, $A$ ($13$ codons)
		and $B$ ($14$ codons). The number arrays indicate the positions of the
		consecutive alignment columns. Codons of
		$A$ and $B$ are colored according to the set to which they belong:
		\texttt{IM} codons in blue color, \texttt{FSext} codons in red color,
		\texttt{InDel} codons in green color and \texttt{FSinit} codons in
		black color. \texttt{MFS} nucleotides contained in \texttt{FSinit} codons
		are underlined.
	}
	\label{fig:CodonSets}
\end{figure}
The following notations and conventions are used in Definition \ref{def:score}
to denote the different sets of codons and nucleotides in $A$ and $B$. 
The set of \texttt{IM} codons in $A$ (resp. $B$) is denoted by
$\texttt{IM}_{A\rightarrow B}$ (resp. $\texttt{IM}_{B\rightarrow A}$). 
The set of \texttt{FSext} codons in $A$ (resp. $B$) is denoted by
$\texttt{FSext}_{A\rightarrow B}$ (resp. $\texttt{FSext}_{B\rightarrow A}$).
The set of \texttt{InDel} codons in $A$ (resp. $B$) is denoted by
$\texttt{InDel}_{A\rightarrow B}$ (resp. $\texttt{InDel}_{B\rightarrow A}$).
The set of \texttt{MFS} nucleotides in $A$ (resp. $B$) is denoted by
$\texttt{MFS}_{A\rightarrow B}$ (resp. $\texttt{MFS}_{B\rightarrow A}$).
In these sets, the codons of $A$ and $B$ are simply identified by the position
(column) of their last nucleotide in the alignment. In this case,
we always have $\texttt{IM}_{A\rightarrow B} =  \texttt{IM}_{B\rightarrow A}$ as
in the example below. The \texttt{MFS} nucleotides are also identified by
their positions in the alignment.

For example, for the alignment depicted in Figure \ref{fig:CodonSets},
the composition of the different sets are:
 $ \texttt{IM}_{A\rightarrow B}  =  \texttt{IM}_{B\rightarrow A} = \{ 3, 9, 12, 15, 26, 48\} $;
 $ \texttt{FSext}_{A\rightarrow B}   =  \{20, 41\} $;
 $ \texttt{InDel}_{A\rightarrow B}  =  \{ 6 \} $; 
  %% FS^-_{A\rightarrow B} =  \{ 23, 29  \}\\
  %% FS^+_{A\rightarrow B}  =  \{ 35, 45\}\\
$ \texttt{MFS}_{A\rightarrow B}  =  \{21,28,29,30,34,35,42,43,45\}$;
 $ \texttt{FSext}_{B\rightarrow A} =  \{ 21, 30, 42\}$;
 $ \texttt{InDel}_{B\rightarrow A}  =  \{ 33\}$; and
  %% FS^-_{B\rightarrow A}  =  \{ 18, 36, 39, 45\}\\
  %% FS^+_{B\rightarrow A}  =  \emptyset\\
 $ \texttt{MFS}_{B\rightarrow A}  =  \{ 18, 34, 35, 39, 43, 45\}.$
 
In the alignment scoring model described in Definition \ref{def:score},
the substitutions of \texttt{IM} and \texttt{FSext} codons  are scored
using an AA scoring function $s_{aa}$ such that aligned codons with
silent nucleotide mutations get the same score as identity. A fixed FS
extension cost denoted by  \texttt{fs\_extend\_cost} is added for
each \texttt{FSext} codon. The insertions/deletions of \texttt{InDel}
codons are scored
by adding a fixed gap cost denoted by \texttt{gap\_cost} for each
\texttt{InDel} codon.
The alignment of \texttt{MFS}  nucleotides are scored independently
from each other, using a nucleotide scoring function $s_{an}$. The insertions
or deletions of nucleotides in \texttt{FSinit} codons are responsible for
the initiation of FS translations. They are then scored by adding a fixed
FS opening cost denoted by \texttt{fs\_open\_cost} for each \texttt{FSinit} codon.
Note that, by convention, the values of all penalty costs for gap and FS (\texttt{gap\_cost}, \texttt{fs\_open\_cost}, \texttt{fs\_extend\_cost}) are negative. Note also that the scoring scheme assumes that the AA and the nucleotide scoring functions, $s_{aa}$ and  $s_{an}$, are symmetric.

\begin{definition}[Score of an alignment] \hfill\\
 \label{def:score}
 Let $(A',B')$ be an alignment of length $L$ between two CDS $A$ and $B$.
  The score of the alignment $(A',B')$ is defined by:

    \[ \begin{array}{lll}
      \texttt{score}(A',B') & =& \sum_{k \in \texttt{IM}_{A\rightarrow B}}{s_{aa}(A'[k-2~..~k],B'[k-2~..~k])} ~ +\\

      & & \sum_{k \in \texttt{FSext}_{A\rightarrow B}}{( \frac{s_{aa}(A'[k-2~..~k],B'[k-2~..~k])}{2} + \texttt{fs\_extend\_cost} )} ~ +\\

      & &  |\texttt{InDel}_{A\rightarrow B}| * \texttt{gap\_cost} ~ +\\

      & & (\frac{|A|}{3} -  |\texttt{IM}_{A\rightarrow B}| - |\texttt{FSext}_{A\rightarrow B}| - |\texttt{InDel}_{A\rightarrow B}|) * \texttt{fs\_open\_cost} ~ + \\

      & & \sum_{k \in \texttt{MFS}_{A\rightarrow B}}{\frac{s_{an}(A'[k],B'[k])}{2}}  ~ + \\
   & & \sum_{k \in \texttt{FSext}_{B\rightarrow A}}{ (\frac{s_{aa}(B'[k-2~..~k],A'[k-2~..~k])}{2} + \texttt{fs\_extend\_cost})} ~+\\
  & &  |\texttt{InDel}_{B\rightarrow A}| * \texttt{gap\_cost} ~+\\
  & & (\frac{|B|}{3} -  |\texttt{IM}_{B\rightarrow A}| - |\texttt{FSext}_{B\rightarrow A}| - |\texttt{InDel}_{B\rightarrow A}|) * \texttt{fs\_open\_cost} ~+ \\  
  & & \sum_{k \in \texttt{MFS}_{B\rightarrow A}}{\frac{s_{an}(B'[k],A'[k])}{2}}\\

    \end{array}
\]
\end{definition}

\section*{Method }
\label{algorithm}

In this section, we describe a $O(nm)$ time and space complexity algorithm
that solves the problem of finding a maximum score alignment between two
CDS $A$ and $B$ of lengths $n$ and $m$.
Similarly to other classical sequence alignment algorithms \cite{needle}, we use
dynamic programming tables that are indexed by the
pairs of prefixes of the two CDS. The table $D$ stores the
maximum scores of the alignments between prefixes of $A$ and $B$. The table
$D_F$ is used to account for potential cases of FS extensions that
are counted subsequently.

\begin{definition}[Dynamic programming tables] \hfill\\
%\subsection*{{Definition 4} [Dynamic programming tables]}
  \label{table}
  Given two CDS $A$ and $B$ as input, the algorithm uses
  two dynamic programming tables $D$ and $D_F$ of size $(n+1)\times (m+1)$.
  The cell $D(i,j)$ contains the maximum score of an alignment between
  the prefixes $A[1~..~i]$ and $B[1~..~j]$.
  The table $D_F$ is filled only for values of $i$ and $j$ such that
  $i (mod~3) = 0$ or $j (mod~3) = 0$.
  If $i (mod~3) \neq 0$
  (resp. $j (mod~3) \neq 0$), the cell $D_F(i,j)$ contains the score of an
  alignment between the prefixes $A[1~..~i+\alpha]$ and $B[1~..~j+\alpha]$ where
  $\alpha = (3-i) (mod~3)$ (resp. $\alpha = (3-j) (mod~3)$).
  The table $D_F$ is filled as follows:
  \begin{itemize}
    \item If $i (mod~3) = 0$ and $j (mod~3) = 0$, $D_F(i,j) = D(i,j)$.
    \item If $i (mod~3) = 0$ and $j (mod~3) = 2$, or $i (mod~3) = 2$
      and $j (mod~3) = 0$, $D_F(i,j)$ contains the maximum score of
      an alignment between $A[1~..~i+1]$ and $B[1~..~j+1]$ such that
      $A[i+1]$ and $B[j+1]$ are aligned together and half of
      the score for aligning $A[i+1]$ with $B[j+1]$ is subtracted.
      
    \item If $i (mod~3) = 0$ and $j (mod~3) = 1$, or $i (mod~3) = 1$ and
      $j (mod~3) = 0$, $D_F(i,j)$ contains the maximum score of an alignment
      between $A[1~..~i+2]$ and $B[1~..~j+2]$ such that $A[i+1]$,$B[j+1]$ and
        $A[i+2]$,$B[j+2]$ are aligned together and half
        of the scores of aligning $A[i+2]$ with $B[j+2]$ and $A[i+1]$ with
        $B[j+1]$ is subtracted. 
 \end{itemize}
\end{definition}
\begin{lemma}[Filling up table D] \hfill\\
%\subsection*{Lemma 1}[Filling up table D]
  \label{D}
  
  \begin{enumerate}
  \item {\bf If $i (mod~3) = 0$ and $j (mod~3) = 0$}
    \scriptsize
      \[ D(i,j) = \max \left\{
  \begin{array}{ll}
    1. & s_{aa}(A[i-2~..~i],B[j-2~..~j]) + D(i-3,j-3)\\
    
    2. & s_{an}(A[i],B[j]) + s_{an}(A[i-1],B[j-1]) + D(i-3,j-2) + 2 * \texttt{fs\_open\_cost}\\
    
    3. & s_{an}(A[i],B[j]) + s_{an}(A[i-2],B[j-1]) + D(i-3,j-2) + 2 * \texttt{fs\_open\_cost}\\
    
    4. & s_{an}(A[i],B[j]) + D(i-3,j-1) + 2* \texttt{fs\_open\_cost}\\
    
    5. & s_{an}(A[i],B[j]) + s_{an}(A[i-1],B[j-1]) + D(i-2,j-3) + 2 * \texttt{fs\_open\_cost}\\
    
    6. & s_{an}(A[i],B[j]) + s_{an}(A[i-1],B[j-2]) + D(i-2,j-3) + 2 * \texttt{fs\_open\_cost}\\

    7. & s_{an}(A[i],B[j]) + D(i-1,j-3) + 2* \texttt{fs\_open\_cost}\\

    8. & s_{an}(A[i],B[j]) + D(i-1,j-1) + 2 * \texttt{fs\_open\_cost}\\

    9. & \frac{s_{an}(A[i-1],B[j])}{2} + \frac{s_{an}(A[i-2],B[j-1])}{2} + D_F(i-3,j-2) + \texttt{fs\_open\_cost}\\
    
    10. & s_{an}(A[i-1],B[j]) + D(i-3,j-1) + 2 * \texttt{fs\_open\_cost}\\
    
    11. & \frac{s_{an}(A[i-2],B[j])}{2} + D_F(i-3,j-1) + \texttt{fs\_open\_cost}\\
    
    12. & \texttt{gap\_cost} + D(i-3,j) \\
    
    13. & D(i-1,j) + \texttt{fs\_open\_cost}\\
    
    14. & \frac{s_{an}(A[i],B[j-1])}{2} + \frac{s_{an}(A[i-1],B[j-2])}{2} + D_F(i-2,j-3) + \texttt{fs\_open\_cost}\\
    
    15. & s_{an}(A[i],B[j-1]) + D(i-1,j-3) + 2 * \texttt{fs\_open\_cost}\\
    
    16. & \frac{s_{an}(A[i],B[j-2])}{2} + D_F(i-1,j-3) + \texttt{fs\_open\_cost}\\
    
    17. & \texttt{gap\_cost} + D(i,j-3) \\
    
    18. & D(i,j-1) + \texttt{fs\_open\_cost}\\  
  \end{array}
\right.
\]

    \normalsize
    \item {\bf If $i (mod~3) = 0$ and $j (mod~3) \neq 0$}

      \scriptsize
      \[ D(i,j) = \max \left\{
  \begin{array}{ll}
  1. &   \frac{s_{aa}(A[i-2~..~i],B[j-2~..~j])}{2} + D_F(i-3,j-3) + \texttt{fs\_extend\_cost}\\
   & + \frac{s_{an}(A[i],B[j])}{2} (+ \frac{s_{an}(A[i-1],B[j-1])}{2} ~if ~j-1 (mod~3) \neq 0)\\
    
  2. &   s_{an}(A[i],B[j]) + s_{an}(A[i-1],B[j-1]) + D(i-3,j-2) + \texttt{fs\_open\_cost} \\
   & (+ \texttt{fs\_open\_cost} ~if ~j-1 (mod~3) = 0)\\

  3. &  s_{an}(A[i],B[j]) + s_{an}(A[i-2],B[j-1]) + D_F(i-3,j-2) + \texttt{fs\_open\_cost}\\
  & (- \frac{s_{an}(A[i-2],B[j-1])}{2} ~if ~j-1 (mod~3) = 0)\\
    
  4. &  s_{an}(A[i],B[j]) + D(i-3,j-1) + \texttt{fs\_open\_cost}\\
      
  5. &  s_{an}(A[i],B[j]) + D(i-1,j-1) + \texttt{fs\_open\_cost}\\

  6. &  s_{an}(A[i-1],B[j]) + s_{an}(A[i-2],B[j-1]) + D_F(i-3,j-2) + \texttt{fs\_open\_cost}\\
   & (- \frac{s_{an}(A[i-2],B[j-1])}{2} ~if ~j-1 (mod~3) = 0)\\
      
  7. & s_{an}(A[i-1],B[j]) + D(i-3,j-1) + \texttt{fs\_open\_cost}\\
    
  8. &  s_{an}(A[i-2],B[j]) + D(i-3,j-1) + \texttt{fs\_open\_cost}\\

  9. &   \texttt{gap\_cost} + D(i-3,j) \\

  10. & D(i-1,j) + \texttt{fs\_open\_cost}\\

  11. &  D(i,j-1)\\
    
  \end{array}
\right.
\]

    \normalsize
      \item {\bf If $i (mod~3) \neq 0$ and $j (mod~3) = 0$}, the equation is symmetric to the previous case.

      \item {\bf If $i (mod~3) \neq 0$ and $j (mod~3) \neq 0$}

      \scriptsize
      $ D(i,j) = \max \left\{
  \begin{array}{ll}
    1. & s_{an}(A[i],B[j]) + D(i-1,j-1)\\

    2. & D(i-1,j)\\
      
    3. & D(i,j-1)\\    
  \end{array}
\right.
$
      \normalsize

  \end{enumerate}
\end{lemma}

The proof of Lemma \ref{D} is given in the Additional file 1.
Figure \ref{fig:D1} illustrates the configurations
  of alignment considered in Lemma \ref{D}
  for computing $D(i,j)$ for Cases 1 and 2.

\begin{figure}[ht!]
\centering
\includegraphics[width=0.85\textwidth]{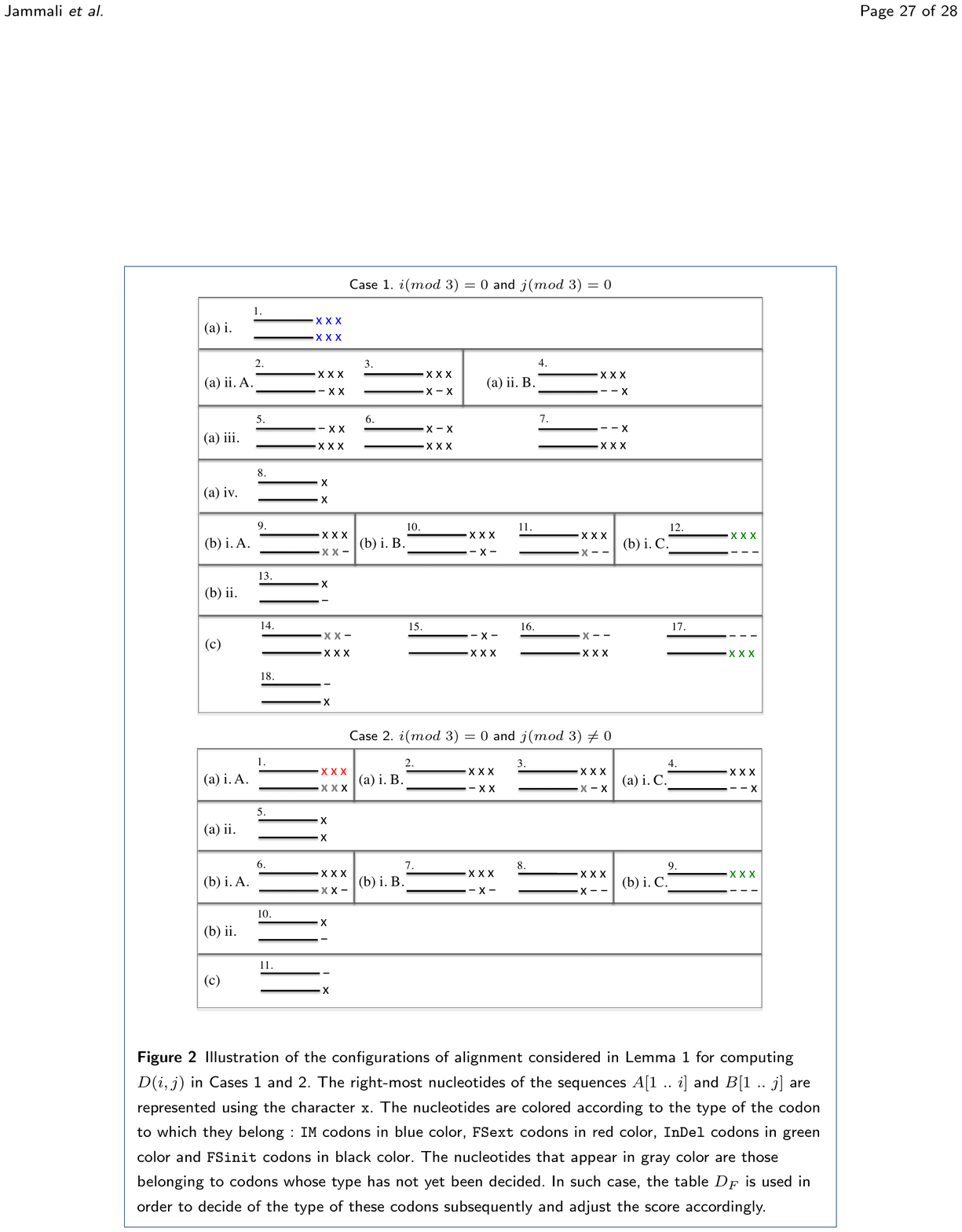}  
\caption{Illustration of the configurations
  of alignment considered in Lemma \ref{D}
  for computing $D(i,j)$ in Cases 1 and 2.
  The right-most nucleotides of the sequences $A[1~..~i]$
  and $B[1~..~j]$ are represented using the character \texttt{x}.
  The nucleotides are colored according to the type of the codon
  to which they belong : \texttt{IM} codons in blue color, \texttt{FSext}
  codons in red color, \texttt{InDel} codons in green color and \texttt{FSinit}
  codons in black color.
  The nucleotides that appear in gray color are those belonging to codons
  whose type has not yet been decided. In such case, the table $D_F$ is used
  in order to decide of the type of these codons subsequently and adjust
  the score accordingly.
}
\label{fig:D1}
\end{figure}
\begin{lemma}[Filling up table $D_F$] \hfill\\
%\subsection*{{Lemma 2}[Filling up table $D_F$]}
    \label{DF}
  \begin{enumerate}
  \item {\bf If $i (mod~3) = 0$ and $j (mod~3) = 0$}\\
    \scriptsize
    $ D_F(i,j) = D(i,j)$\\
    
    \normalsize
  \item {\bf If $i (mod~3) = 2$ and $j (mod~3) = 0$}\\
    \scriptsize    
      \[ D_F(i,j) = \max \left\{
  \begin{array}{ll}
    1. & \frac{s_{aa}(A[i-1~..~i+1],B[j-1~..~j+1])}{2} + D_F(i-2,j-2) + \texttt{fs\_extend\_cost}\\
    
    2. & \frac{s_{an}(A[i+1],B[j+1])}{2} + s_{an}(A[i],B[j]) + D(i-2,j-1) + 2 * \texttt{fs\_open\_cost}\\
    
    3. & \frac{s_{an}(A[i+1],B[j+1])}{2} + \frac{s_{an}(A[i-1],B[j])}{2} + D_F(i-2,j-1) + \texttt{fs\_open\_cost}\\

    4. & \frac{s_{an}(A[i+1],B[j+1])}{2}  + D(i-2,j) + \texttt{fs\_open\_cost}\\

    5. & \frac{s_{an}(A[i+1],B[j+1])}{2}  + D(i,j) + \texttt{fs\_open\_cost}\\
  \end{array}
\right.
\]
    \normalsize
  \item {\bf If $i (mod~3) = 0$ and $j (mod~3) = 2$}, the equation is symmetric to the previous case.\\

\item {\bf If $i (mod~3) = 1$ and $j (mod~3) = 0$}\\
    \scriptsize    
      \[ D_F(i,j) = \max \left\{
  \begin{array}{ll}
    1. & \frac{s_{aa}(A[i~..~i+2],B[j~..~j+2])}{2} + D_F(i-1,j-1) + \texttt{fs\_extend\_cost}\\
    
    2. & \frac{s_{an}(A[i+2],B[j+2])}{2} + \frac{s_{an}(A[i+1],B[j+1])}{2}  + D(i-1,j) + \texttt{fs\_open\_cost}\\
    
    3. & \frac{s_{an}(A[i+2],B[j+2])}{2} + \frac{s_{an}(A[i+1],B[j+1])}{2}  + D(i,j) + \texttt{fs\_open\_cost}\\
  \end{array}
\right.
\]
    \normalsize

  \item {\bf If $i (mod~3) = 0$ and $j (mod~3) = 1$}, the equation is symmetric to the previous case.
  \end{enumerate}
\end{lemma}

\begin{proof}[Proof of Lemma \ref{DF}.] The proof follows from Lemma \ref{D}.
  \begin{enumerate}
  \item {\bf If $i (mod~3) = 0$ and $j (mod~3) = 0$}, this case is trivial.
  \item {\bf If $i (mod~3) = 2$ and $j (mod~3) = 0$}, then $i+1 (mod~3) = 0$
    and $j+1 (mod~3) = 1 \neq 0$. The five cases follow from the
    application of Lemma \ref{D}, Case 2 for computing $D(i+1,j+1)$,
    and by keeping
    only the cases where $A[i+1]$ and $B[j+1]$ are aligned together (cases 1,
    2, 3, 4, 5 among the 11 cases). However, in each of the cases,
    we must subtract half of the score of aligning $B[i+1]$ with $A[j+1]$
    ($\frac{s_{an}(A[i+1],B[j+1])}{2}$), because this score will be added
    subsequently.
    
  \item {\bf If $i (mod~3) = 0$ and $j (mod~3) = 2$}, the proof is symmetric to the previous case.

  \item {\bf If $i (mod~3) = 1$ and $j (mod~3) = 0$}, then $i+2 (mod~3) = 0$
    and $j+2 (mod~3) = 2 \neq 0$. Here again, the three cases follow
    from the application of Lemma \ref{D}, Case 2 for computing
    $D(i+2,j+2)$ and
    by keeping only the cases where $A[i+1]$, $B[i+1]$ and $A[i+2]$,
    $B[i+2]$ can be aligned together (cases 1, 2, 5 among the 11 cases).
    However, in each of the cases, we must subtract half of the
    scores of aligning $B[i+2]$ with $A[j+2]$ and aligning $B[i+1]$ with
    $A[j+1]$ ($\frac{s_{an}(A[i+2],B[j+2])}{2}$,
    $\frac{s_{an}(A[i+1],B[j+1])}{2}$), because theses scores will be added
    subsequently.
    
  \item {\bf If $i (mod~3) = 0$ and $j (mod~3) = 1$}, the proof is symmetric to the previous case.
  \end{enumerate}
\end{proof}

The alignment algorithm using Lemma \ref{D}  and  \ref{DF}
is described in the next  theorem. 

\begin{theorem}[Computing a maximum score alignment] \hfill\\
  \label{thm}
  Given two CDS $A$ and $B$ of lengths $n$ and $m$,
  a maximum score alignment between $A$ and $B$ can be computed in
  time and space  $O(nm)$, using the following algorithm.
  \scriptsize    
  \noindent
  $\texttt{Algorithm Align(A,B)}$\\
  $\texttt{~~for i = 0 to n do}$\\
      $\texttt{~~~~~~}D(i,0) = floor(\frac{i}{3}) * \texttt{gap\_cost}$\\
      $\texttt{~~~~~~}D_F(i,0) = D(i,0) + \left\{
   \begin{array}{ll}
    \frac{s_{an}(A[i+1],B[1])}{2} + \frac{s_{an}(A[i+2],B[2])}{2} + \texttt{fs\_open\_cost}, &  \texttt{if ~i (mod~3) = 1}\\
    \frac{s_{an}(A[i+1],B[1])}{2}  + \texttt{fs\_open\_cost}, &  \texttt{if ~i (mod~3) = 2}
  \end{array}
   \right.$
   
  \noindent   
  $\texttt{~~for j = 0 to m do}$\\
      $\texttt{~~~~~~}D(0,j) = floor(\frac{j}{3}) * \texttt{gap\_cost}$\\
      $\texttt{~~~~~~}D_F(0,j) = D(0,j) + \left\{
   \begin{array}{ll}
    \frac{s_{an}(A[1],B[j+1])}{2} + \frac{s_{an}(A[2],B[j+2])}{2} + \texttt{fs\_open\_cost}, &  \texttt{if ~j (mod~3) = 1}\\
    \frac{s_{an}(A[1],B[j+1])}{2}  + \texttt{fs\_open\_cost}, &  \texttt{if ~j (mod~3) = 2}\\
  \end{array}
\right.$
  $\texttt{~~for i = 0 to n do}$\\
      $\texttt{~~~~~~for j = 0 to m do}$\\
 $\texttt{~~~~~~~~~~compute D(i,j) using Lemma \ref{D}}$\\
 $\texttt{~~~~~~~~~~compute $D_F$(i,j) using Lemma \ref{DF}, ~if ~i (mod~3) = 0 ~or ~j (mod~3) = 0}$
 
\end{theorem}
\normalsize

\begin{proof}[Proof of Theorem \ref{thm}] \hfill \\
  The proof relies on two points: (1) The algorithm computes the maximum
  score of an alignment between $A$ and $B$ and (2) the algorithm runs
  with an $O(nm)$ time and space complexity.\\

  \noindent
  (1) The validity of the algorithm, i.e. the fact that it fills the cells of
  the tables $D$ and $D_F$ according to Definition \ref{table}, follows from
  five points.
  \begin{itemize}
  \item The initialization of the tables is a direct consequence of
  Definition \ref{table}.
  \item Lemmas \ref{D} and \ref{DF}.
  \item The couples $(i,j)$ of prefixes of $A$ and $B$ that need to be
    considered
    in the algorithm are all the possible couples for $D(i,j)$ and only
    the couples such
  that $i (mod~3) = 0$ or $j (mod~3) = 0$ for $D_F(i,j)$ (see all the cases
  in which the table $D_F$ is used in Lemmas \ref{D} (7 cases) and \ref{DF}
  (3 cases)).
  \item The couples $(i,j)$ of prefixes of $A$ and $B$ are considered
  in increasing
  order of length and $D[i,j]$ is computed before $D_F[i,j]$ in the cases
  where $i (mod~3) = 0$ or $j (mod~3) = 0$.
  \item A backtracking of the algorithm allows to find a maximum
  score alignment between $A$ and $B$.
  \end{itemize}

  \noindent
  (2) The time and space complexity of the algorithm is a direct consequence
  of the number of cells of the tables $D$ and $D_F$,
  $2 \times (n+1) \times (m+1)$. Each cell is filled in constant time.
The exact formula for the computational complexity of the algorithm is computed below.\\
\begin{tabular}{lllll}
& $18$ & $\times$ & $\frac{nm}{9}$ & ~~~~~for $\frac{nm}{9}$ calls of the Case 1 of Lemma \ref{D}\\
+& $11$ & $\times$ & $2\times \frac{nm}{3}$ & ~~~~~for $2\times \frac{nm}{3}$ calls of the Cases 2 or 3 of Lemma \ref{D}\\
+& $3$ & $\times$ & $\frac{4nm}{9}$ & ~~~~~for $\frac{4nm}{9}$ calls of the Case 4 of Lemma \ref{D}\\
+& $1$ & $\times$ & $\frac{nm}{9}$ & ~~~~~for $\frac{nm}{9}$ calls of the Case 1 of Lemma \ref{DF}\\
+& $5$ & $\times$ & $2 \times \frac{nm}{9}$ & ~~~~~for $2\times \frac{nm}{9}$ calls of the Cases 2 and 3 of Lemma \ref{DF}\\
+& $3$ & $\times$ & $2 \times \frac{nm}{9}$ & ~~~~~for $2\times \frac{nm}{9}$ calls of the Cases 4 and 5 of Lemma \ref{DF}\\
Total & = &\multicolumn{3}{l}{$12.55nm$}  \\
\end{tabular}
\\
  \end{proof}

\section*{Results and discussion}

 We implemented the present CDS alignment algorithm with an affine gap penalty scheme \cite{altschul1986optimal}
 such that the penalty for a concatenation of $k$ inserted (resp. deleted)
 codons is
 $\texttt{gap\_open\_cost} + k * \texttt{gap\_cost}$, such that $\texttt{gap\_open\_cost}$ is a negative penalty cost for gap initiations.
This was done by
 adding two dynamic programming tables $G_A$ and $G_B$ such that 
 the cell $G_A(i,j)$ (resp. $G_B(i,j)$) contains the maximum score of an
 alignment between the prefixes $A[1~..~i]$ and $B[1~..~j]$ where the codon
 $A[i-2~..~i]$ (resp. $B[j-2~..~i]$) is an \texttt{InDel} codon.

\subsection*{\bf Data}
\label{experiment}
 We evaluated the algorithm through applications on a mammalian dataset
 containing CDS sequences from ten gene families obtained from the database
 Ensembl-Compara version 83 \cite{cunningham2015ensembl}. The first gene family named "FAM86" is such that three CDS from three of its paralogous human genes were shown in \cite{okamura2006frequent} to share a common FS region translated in three different frames in the three CDS (see Figure \ref{fig:fam86-1} for an illustration of the multiple alignment of these three CDS). The nine other families are  the nine smallest (in term of the overall length of CDS) of fifteen gene families listed in \cite{raes2005functional} where they were shown to display 
 one FS translation region between some pairs of CDS.
 For each gene family, the CDS  of all \emph{human}, \emph{mouse} and
 \emph{cow} genes belonging to the family and satisfying Definition \ref{def:CDS} were downloaded. The overall number
 of distinct pairs of CDS within the ten gene families is $4011$.
 Table \ref{tab:tenfamiliesdetails} gives the details about the content and size of the ten gene
 families (The CDS of the ten gene families are provided in the Additional file 2). 

\begin{figure}[ht!]
	\centering		
	\includegraphics[width=0.85\textwidth]{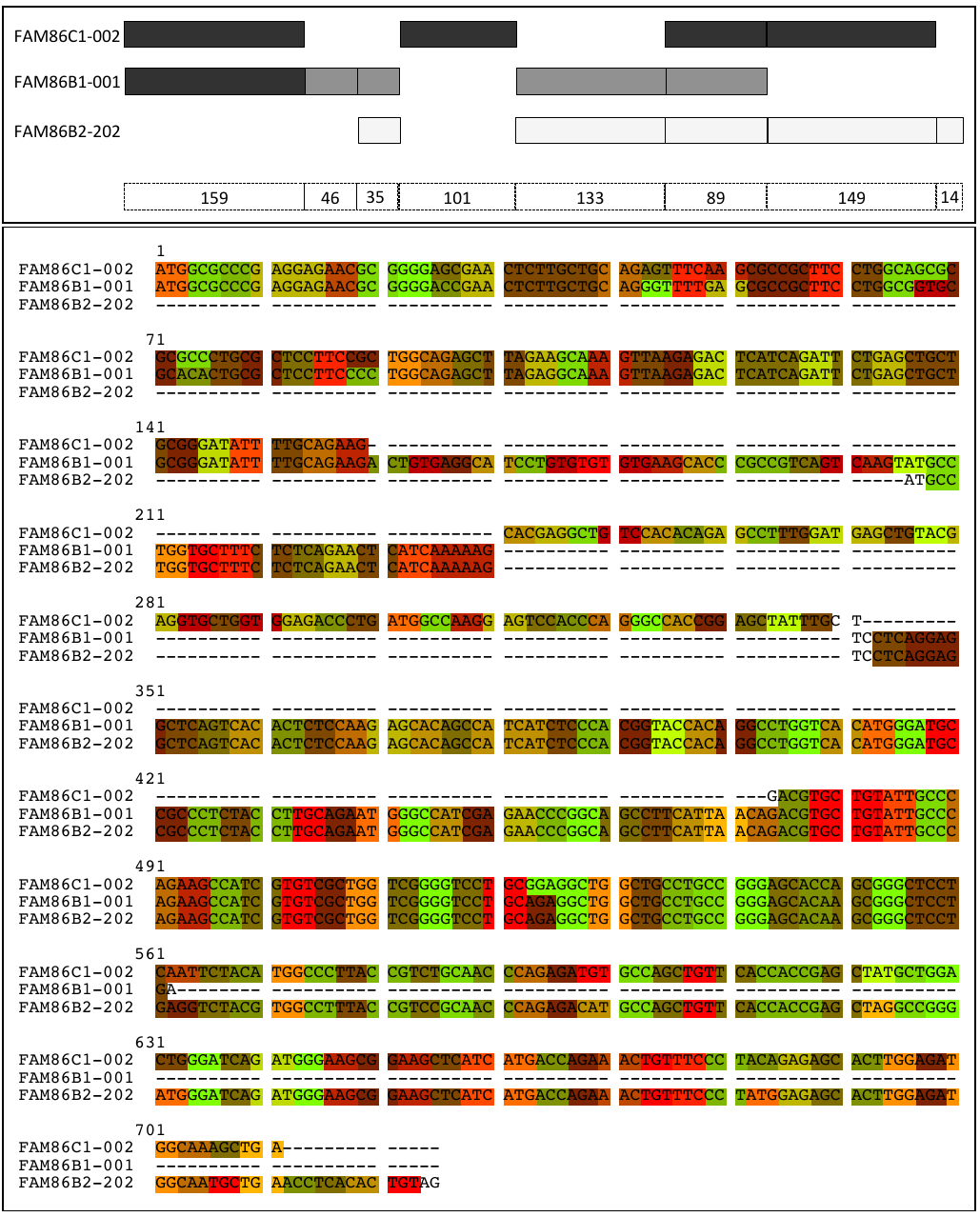}
	
 	\caption{ {\bf Top.} Rough representation of the real alignment of CDS \emph{FAM86C1-002},
 		\emph{FAM86B1-001} and \emph{FAM86B2-202}. Rectangular colored portions represent
 		concatenations of nucleotides in the alignment while blank portions represent
 		concatenations of gap symbols. The lengths of the alignment portions are given
 		at the bottom. 
 		The colors of the nucleotide regions indicate
 		the coding frame in which they are translated, taking the frame of
 		CDS \emph{FAM86C1-002} as reference.
 		For example, there is a nucleotide
 		region of length $89$ shared by the three CDS and translated in $3$
 		different coding frames. {\bf Bottom.} Real alignment of three CDS (Figure obtained using the visualization software \texttt{seaview} \cite{gouy2010seaview}). Nucleotides are colored according to the codon structure of the first CDS \emph{FAM86C1-002}.
 	}
 	\label{fig:fam86-1}
 \end{figure}

\begin{table}[ht!]
\flushleft
  \caption{Detailed description of the ten gene families of the mammalian dataset.}
\begin{tabular}{|l|c|c|c|l|l|}
\hline
Gene family & Human gene & \# of genes & \# of CDS & Length & $\frac{N*(N-1)}{2}$ \\
\hline
I (FAM86) & ENSG00000118894 & 6 & 14 & 10335 & 91 \\
II (HBG017385) & ENSG00000143867 & 6 & 10 & 8988 & 45\\
III (HBG020791) & ENSG00000179526 & 6 & 10 & 11070 & 45\\
IV (HBG004532)  & ENSG00000173020 & 17 & 33 & 52356 & 528\\
V (HBG016641) & ENSG00000147041 & 13 & 33 & 64950 & 528\\
VI (HBG014779)  & ENSG00000233803 & 28 & 44 & 45813 & 946\\
VII (HBG012748)  & ENSG00000134545 & 24 & 44 & 28050 & 946\\
VIII (HBG015928) & ENSG00000178287 & 5 & 19 & 5496 & 171\\
IX (HBG004374) & ENSG00000140519 & 13 & 30 & 36405 & 435\\
X (HBG000122) & ENSG00000105717 & 11 & 24 & 27081 & 276\\
\hline
\multicolumn{5}{|l|}{Total number of pairs of CDS} & 4011 \\
\hline
\end{tabular}
\\
For each gene family, the family identifier used in \cite{okamura2006frequent} or \cite{raes2005functional}, the Ensembl identifier of a \emph{human} gene member
    of the family,
    the number of \emph{human}, \emph{mouse} and \emph{cow} genes in the family,
    the total number of CDS of these genes, the total sum of lengths of
    these CDS and the number of distinct pairs of CDS are given.
\label{tab:tenfamiliesdetails}
\end{table}

\subsection*{\bf Evaluation strategies} 

We compared the accuracy of five pairwise global alignment methods, including the present method, for computing
CDS alignments in the presence or absence of FS translation
between the compared CDS.
The five
methods vary according to the alignment algorithm used, either the present
CDS alignment algorithm called FsePSA allowing
to penalize both FS translation initiation and extension, or the  CDS 
alignment algorithm called MACSE \cite{ranwez2011} penalizing FS
translation initiation, or the Needleman-Wunsch (NW) sequence alignment
algorithm \cite{needle} penalizing neither. Table \ref{tab:fivemethods} 
summarizes the alignment algorithm and the values of parameters used
for each of the five methods.

\begin{table}[ht!]
\flushleft
  \caption{Description of the five methods considered in the experiment.}
\begin{tabular}{|l|l|l|l|}
\hline
\texttt{Method} & \texttt{Alignment approach \&} & \texttt{FS initiation cost} & \texttt{Other parameters}\\
 & \texttt{specific parameters} &  & \\
\hline
\multirow{2}{*}{\texttt{fse}} & Present approach &  \multirow{3}{*}{$\texttt{fs\_open\_cost} = $} & $\texttt{AA gap\_open\_cost} = -11$ \\ 
& $\texttt{fs\_extend\_cost} = $ & & $\texttt{AA gap\_cost} = -1$\\
&  $-1; -0.5; -0.2$ & -10; -20; -30 & $s_{aa} = \texttt{BLOSUM62 matrix}$\\\cline{1-2}
\multirow{2}{*}{\texttt{fse0}} & Present approach & &  $s_{an} =$ \texttt{+1/-1}   \\
& $\texttt{fs\_extend\_cost} = 0$ & &  \texttt{match/mismatch}\\\cline{1-2}  
\multirow{2}{*}{\texttt{macse\_p}} & Ranwez et al. \cite{ranwez2011} & &  \\ 
& $\texttt{stop\_cost} = -100$ & & \\\cline{1-3}
\texttt{needleprot} & NW \cite{needle} at AA level & not applicable & \\ \cline{1-4}
\multirow{4}{*}{\texttt{needlenuc}} & \multirow{4}{*}{NW \cite{needle} at NT level} & \multirow{4}{*}{not applicable} & $\texttt{NT gap\_open\_cost} = -5$\\
 &  &  &  $\texttt{NT gap\_cost} = -2$\\
&  &  & $s_{an} =$ \texttt{+2/-3} \\
&  &  & \texttt{match/mismatch}\\ \cline{1-4}
\hline
\end{tabular}
\\
For each method, the alignment approach and the values
    of specific and common parameters are given.
\label{tab:fivemethods}
\end{table}

The present CDS alignment algorithm is used in two
of the five methods, namely \texttt{fse} and \texttt{fse0}. These two
methods differ according to the value given
to the parameter $\texttt{fs\_extend\_cost}$, 
either $\texttt{fs\_extend\_cost < 0}$ ($-1$, $-0.5$ or  $-0.2$)
for the method
\texttt{fse} penalizing FS translation extension, or
$\texttt{fs\_extend\_cost} = 0$ for the method
\texttt{fse0} not penalizing FS translation extension.
The pairwise version of MACSE \cite{ranwez2011} is used in the 
method called \texttt{macse\_p}.
The NW alignment algorithm is used in the last two methods,
the method called \texttt{needlenuc} computing scores and
alignments at the nucleotide level and the method called \texttt{needleprot}
at the AA level.
For all methods using both the amino acid
and nucleotide scoring functions $s_{aa}$ and $s_{an}$,
$s_{an}$ was fixed to \texttt{+1/-1}
for \texttt{match/mismatch}, so that the overall score of $3$ consecutive
nucleotide identities in an alignment scores less than the smallest
identity score in $s_{aa}$.
All other parameters shared by several methods were given the same value
for all methods. In particular, for the three methods \texttt{fse},
\texttt{fse0} and \texttt{macse\_p} penalizing FS translation initiation,
the parameter $\texttt{fs\_open\_cost}$ was given the values $ -10$, $-20$ or $-30$. 
All other parameters were fixed to the default values for the
NW algorithm implementation of NCBI Blast at the nucleotide and AA
levels \cite{ncbi-blast}. 

We used the five methods to compute pairwise alignments between the
pairs of CDS within each of the ten gene families of our dataset,
yielding $4011$ alignments in total for each of five methods.
In the absence of available benchmarks for the direct evaluation of the accuracy of
CDS alignments, we base our evaluation on four indirect strategies. 

In the first strategy, we consider the CDS multiple alignment of each gene family obtained using MACSE \cite{ranwez2011} as a benchmark. This strategy exploits the fact that multiple alignments are usually more accurate than pairwise alignments. It then assumes that the MACSE multiple alignments are closer to the reality than the pairwise alignments obtained using the five methods.
Note that all the pairwise alignment methods included in the comparison can be extended to multiple sequence alignment methods using classical strategies. Thus, the more accurate pairwise alignment methods should lead to more accurate multiple alignment methods. Here, we 
focus on the comparison of the pairwise versions of the methods.
%
%(a) to count for frameshifts opening and (b) to compare the nucleotides positions in an aligned sequence; 
In the second strategy, we consider six composition criteria for a CDS pairwise alignment called \texttt{Identity\_NT}, \texttt{Identity\_AA},  \texttt{Gap\_init},
\texttt{Gap\_length}, \texttt{FS\_init}, \texttt{FS\_length}. The definitions of these criteria are given below, and used to compare the five methods. In the third strategy, we manually build and use as a benchmark, the real multiple alignment of three CDS from three paralogous human genes of the gene family I (FAM86). In the fourth strategy, we generate and use a set of three CDS splicing orthology groups, each group containing seven existing or putative CDS from seven genes of gene family I (FAM86).

Based on the results of the large-scale experiments discussed in the following, 
the best compromise for default values of FsePSA parameters are $-30$ for 
\texttt{fs\_open\_cost} and $-1$ for \texttt{fs\_extend\_cost}.

\subsection*{\bf Discussion}
\subsubsection*{\bf First strategy: Using MACSE multiple alignments as benchmark}
MACSE \cite{ranwez2011} was used with its default parameters (\texttt{fs\_open\_cost} = $-30$,
\texttt{stop\_cost} = $-100$, $s_{aa} = \texttt{BLOSUM62 matrix}$, \texttt{gap\_open\_cost} = $-7$,
\texttt{gap\_cost} = $-1$) to compute the CDS multiple alignment of each of the ten gene families. For each MACSE multiple alignment of $N$ CDS, we consider the $\frac{N(N-1)}{2}$ induced pairwise alignments as a benchmark. In total, we then obtained a benchmark composed of $4011$ pairwise alignments. In order to compare an alignment $(A',B')$ obtained with one of the five methods to the corresponding alignment $(A",B")$ in the benchmark, we computed the number of nucleotides aligned in $(A',B')$ with the same partner as in the benchmark alignment $(A",B")$.

Table \ref{tab:scorecomparison} shows the overall percentage of nucleotides aligned with the same partners as in the benchmark for each of the compared methods, for varying \texttt{fs\_open\_cost} ($-10$, $-20$ and $-30$) and \texttt{fs\_extend\_cost} ($-1$, $-0.5$ and $-0.2$). It shows that the different versions of the \texttt{fse} method and the \texttt{fse0} method have the best scores greater than $79.4\%$, followed by the \texttt{needleprot} method with a score of $78.82\%$. On the opposite, the \texttt{needlenuc} and \texttt{macse\_p} method with \texttt{$fs\_open\_cost $}= $-30$ return the worst scores, respectively $50.95\%$ and $47.35\%$. These results also show that the \texttt{fse} method is more
robust to the \texttt{fs\_open\_cost} parameter changes
as compared to the \texttt{macse\_p} method, whose scores show a large variation between $47.35\%$ and $78.29\%$. Note that the \texttt{needlenuc} and \texttt{needleprot}
do not account for the \texttt{fs\_open\_cost} parameter.

\begin{landscape}
\begin{table}[ht!]
  \flushleft
	\caption{Comparison with MACSE multiple alignments benchmark. }
	\begin{tabular}{|l|l|l|l|l|l|l|l|}
		\hline
		\texttt{fs\_open\_cost} & \texttt{fse0} & \begin{tabular}[c]{@{}l@{}}\texttt{fse}\\ (-1)\end{tabular} & \begin{tabular}[c]{@{}l@{}}\texttt{fse}\\ (-0.5)\end{tabular} & \begin{tabular}[c]{@{}l@{}}\texttt{fse}\\ (-0.2)\end{tabular}& \texttt{macse\_p} &  \texttt{needlenuc} &
		 \texttt{needleprot}  \\
		\hline	
		\texttt{-10} & \textbf{79.58 (1404)} & 79.40 (1364) & \textbf{79.52 (1415) }& \textbf{79.58 (1433)} & 77.17 (1076) &  \multirow{3}{*}{50.95 (255)} &
		\multirow{3}{*}{78.82 (972)}   \\
		\cline{1-6} 		

		\texttt{-20} & $\textbf{79.68 (1550)}$ & $\textbf{79.68 (1526)}$ & $\textbf{79.65 (1558)}$ & $\textbf{79.67 (1552)}$ & 78.29 (1389) &   &
		  \\
		\cline{1-6}
		\texttt{-30} & $\textbf{79.75 (1558)}$ & $\textbf{79.47 (1529)}$ & $\textbf{79.60 (1546)}$ & $\textbf{79.63 (1547)}$ & 47.35 (742) &   &
		  \\
		\hline
%    \multicolumn{2}{l|}{Total} & \multicolumn{6}{l|}{9022314}  \\
%\hline
	\end{tabular}
 
        Percentage of nucleotides aligned with the same partner as in the benchmark alignments induced by the MACSE multiple alignments, for each method for varying \texttt{fs\_open\_cost} ($-10$, $-20$ and $-30$) and \texttt{fs\_extend\_cost} ($-1$, $-0.5$ and $-0.2$). In each case, the
number of CDS pairs with an alignment that presents the highest similarity with the corresponding benchmark alignment as compared to the other methods is given in parenthesis. The best results are indicated in bold.
\label{tab:scorecomparison}
\end{table}

\begin{table}[ht!]
 \flushleft
\caption{Values of the six criteria for the \texttt{\bf noFS} dataset (variations as compared to \textbf{needleprot}).}
\begin{tabular}{|l|l|l|l|l|l|l|l|}
\hline
\begin{tabular}[c]{@{}l@{}}\texttt{fs\_open\_cost}\\\texttt{(\# CDS pairs)} \end{tabular} & Method     & \begin{tabular}[c]{@{}l@{}}\texttt{Identity}\_\\\texttt{NT} \end{tabular} & \begin{tabular}[c]{@{}l@{}}\texttt{Identity}\_\\\texttt{AA} \end{tabular} & \begin{tabular}[c]{@{}l@{}}\texttt{Gap}\_\\\texttt{init}\end{tabular} & \begin{tabular}[c]{@{}l@{}}\texttt{Gap}\_\\\texttt{length}\end{tabular} & \begin{tabular}[c]{@{}l@{}}\texttt{FS}\_\\\texttt{init}\end{tabular} & \begin{tabular}[c]{@{}l@{}}\texttt{FS}\_\\\texttt{length}\end{tabular} \\ \hline
    \hline

    \multirow{5}{*}{-10 (1672)}  & \texttt{fse0}     & \multirow{2}{*}{3281 (1158)} & 	\multirow{2}{*}{5376 (1222)}& 	\multirow{2}{*}{-495 (1606)}& 	\multirow{2}{*}{-2718 (1521)}& 	\multirow{2}{*}{0 (1672)}& 	\multirow{2}{*}{0 (1672)}  \\\cline{2-2}
  & \texttt{fse}     & & 	& 	& 	& 	&   \\ \cline{2-8}
    &     \texttt{macse\_p}  & 8120 (955)&	27942 (676)	&3701 (711)	&9618 (1102)&	0 (1672)&	0 (1672)
  \\ \cline{2-8} 
    &     \texttt{needlenuc} & 170239 (156)&	-82002(442)&	104811 (218)&	21422(427)&	44488 (256)&263365 (256) \\ \cline{2-8} 
    &  \textbf{needleprot} & \textbf{1090957}      & \textbf{2047608}      & \textbf{10230}     & \textbf{530688 }     & \textbf{0}        & \textbf{0}  \\
  \hline
  \hline
    \multirow{5}{*}{-20 (3441)}  & \texttt{fse0}     & \multirow{2}{*}{1409 (2612)}& 	\multirow{2}{*}{-8622 (2672)}& 	\multirow{2}{*}{-3564 (3169)}& 	\multirow{2}{*}{-9984 (3057)}& 	\multirow{2}{*}{0 (3441)}& 	\multirow{2}{*}{0 (3441)}  \\\cline{2-2}
  & \texttt{fse}     & & 	& 	& 	& 	&   \\ \cline{2-8}

  &\texttt{macse\_p}  & 24909	(1437)&95844 (1011)&	13778 (1076)&	30884 (1791)&	0	(3441)&0 (3441)
   \\ \cline{2-8}  
  & needleenuc & 547203 (176)&	-177285(680)	&317256 (219)	&52510 (552)	&138204 (257)	&844401 (257)
  \\  \cline{2-8}
&\textbf{needleprot}    & \textbf{2000228}      & \textbf{3494760 }     & \textbf{31793 }    & \textbf{1313658}     & \textbf{0}        & \textbf{0}

  \\
  \hline
  \hline
    \multirow{5}{*}{-30 (3740)}  & \texttt{fse0}     & \multirow{2}{*}{1368 (2834)}& 	\multirow{2}{*}{-10788 (2912)}& 	\multirow{2}{*}{-4047 (3448)}& 	\multirow{2}{*}{-11316 (3321)}& 	\multirow{2}{*}{0 (3740)}& 	\multirow{2}{*}{0 (3740)}  \\\cline{2-2}
  & \texttt{fse}     & & 	& 	& 	& 	&   \\ \cline{2-8}

    &     \texttt{macse\_p}  & 27840 (1547)&	106512 (1078)&	15561 (1117)&	34726  (1846)&	0 (3740)&	0 (3740)
  \\ \cline{2-8} 
   & \texttt{needlenuc} & 610305 (177)&	-192231	(709)&351748 (219)&	47356 (573)&	154255 (257)&	948418 (257)
     \\ \cline{2-8}
& \textbf{needleprot}   &\textbf{ 2143630} & \textbf{3715632}      & \textbf{35296 }    & \textbf{1439784}     &\textbf{ 0 }       & \textbf{0 }   
      \\ \hline
\end{tabular}
\\
For varying values of the parameter \texttt{fs\_open\_cost}, the number of CDS pairs in the dataset is given. The values of the criteria for the reference method ``\texttt{needleprot}'' are indicated in bold characters. For each of the other methods (\texttt{fse}, \texttt{fse0}, \texttt{macse\_p}, \texttt{needlenuc}), the variations of the criteria values as compared to the reference values are given.
For each criteria and each method, the number of CDS pairs that have the closest value to the reference \texttt{needleprot} value is given in parentheses.
\label{tab:nofs} 
\end{table}
\end{landscape}

\subsubsection*{\bf Second strategy: Using six composition criteria for CDS pairwise alignment}

Six criteria were defined and used to compare the five pairwise alignment methods. 
Given a pairwise CDS alignment, the first criterion \texttt{Identity\_NT} counts 
the number of gap-free columns in
the alignment containing a nucleotide match. The second criterion \texttt{Identity\_AA} counts
the number of \texttt{IM} and \texttt{FSext} codons $c$ in the alignment that
are aligned with a triplet of nucleotides yielding the same amino acid as $c$.
The third criterion \texttt{Gap\_init} is the number of gap-containing columns in the
alignment, either insertion or deletion columns that are preceded by a
different type of column.
The fourth criterion \texttt{Gap\_length} is the overall number of gap-containing columns
in the alignment.
The fifth criterion \texttt{FS\_init} is the number of FS translation segments found in the alignment.
The last criterion \texttt{FS\_length} is
the overall number of columns in the alignment intersecting a \texttt{FSext} codon.   

Note that the definitions of the six criteria exploit the definitions
of codon sets used in Definition \ref{def:score} but they are independent of
any alignment scoring scheme.
For example, for the alignment depicted in Figure \ref{fig:CodonSets},
$\texttt{Identity\_NT}=28$, counting all gap-free columns except the
five columns at the positions $\{9,12,15,42,45\}$ containing
a nucleotide mismatch.
$\texttt{Identity\_AA}=14$, counting all \texttt{IM} and \texttt{FSext}
codons except the two \texttt{IM} codons \texttt{AAG} and \texttt{AAT}
ending at position $15$ yielding two different amino acids $K$ and $N$,
and the \texttt{FSext} codon \texttt{AAT} ending at position $42$
yielding the amino acid $N$ different from the amino acid $K$ yielded
by the triplet \texttt{AAG}. $\texttt{Gap\_init} = 7$, counting the
positions $\{4,16,22,27,31,36,44\}$.
$\texttt{FS\_init} = 3$, counting the positions $\{18,28,39\}$.
The two last criteria have the values $\texttt{Gap\_length} = 15$
and $\texttt{FS\_length} = 11$.

For each of the nine cases obtained by combining the values of the parameters
\texttt{fs\_open\_cost} ($-10$, $-20$ or $-30$) and \texttt{fs\_extend\_cost} ($-1$, $-0.5$ or $-0.2$), we considered the $4011$ pairs of CDS from the ten gene families dataset, and partitioned them into three sets.  For each case, the first set called the \texttt{noFS} dataset is composed of the pairs of CDS for which the pairwise alignments obtained using the \texttt{fse0}, \texttt{fse} and \texttt{macse\_p} methods all have the criteria $\texttt{FS\_init} = 0$. 
The second set called the \texttt{FS} dataset is composed of the pairs of CDS for which the alignments obtained using the \texttt{fse0}, \texttt{fse} and \texttt{macse\_p} methods all have the criteria $\texttt{FS\_init} > 0$. The third set called  the \texttt{ambiguFS} dataset is composed of the remaining pairs of CDS.

Note that, in all nine cases, the set of CDS pairs
for which $\texttt{FS\_init} = 0$ with the \texttt{macse\_p} method was strictly
included in the set of CDS pairs for which $\texttt{FS\_init} = 0$ with
the \texttt{fse} method. 
For each of the nine cases, we computed the overall value of the six criteria for each  method (\texttt{fse0}, \texttt{fse}, \texttt{macse\_p}, \texttt{needlenuc} and \texttt{needleprot}) and each dataset (\texttt{noFS}, \texttt{FS} and \texttt{ambiguFS}). Tables
\ref{tab:nofs}, \ref{tab:fsonly} and \ref{tab:fsambigu} present the results.

\paragraph{\bf Results for the \texttt{noFS} datasets.} For the \texttt{noFS} datasets, we assume
that the real alignments should not contain FS translations. So, the \texttt{needleprot} method most likely
computes the more accurate alignments since it does not allow any FS
translation in the alignments. Indeed, it computes a maximum score NW alignment
at the AA level and back-translates this alignment at the nucleotide level.
We then take the \texttt{needleprot} result as a reference for the
\texttt{noFS} dataset, in all cases. 
By construction of the \texttt{noFS} dataset, for  a fixed value of the parameter \texttt{fs\_open\_cost}, the \texttt{fse0} and \texttt{fse}
methods necessarily return two alignments with the same similarity score for each pair of CDS of the dataset. Indeed, we observed that, for each value of \texttt{fs\_open\_cost} ($-10$, $-20$ or $-30$), the alignments obtained using the methods \texttt{fse0} or \texttt{fse} with varying values of the parameter \texttt{fs\_extend\_cost} are unchanged.

Table \ref{tab:nofs} summarizes 
the results for $\texttt{fs\_open\_cost}=-10$, $-20$ and $-30$, presenting the results of the varying versions of \texttt{fse} and \texttt{fse0} in a single line in the $3$ cases. 
It shows that the results of the
\texttt{fse} and \texttt{fse0} methods are the closest to the reference for all the six criteria in all cases.
However, they slightly overestimate or underestimate the criteria.
The tendency of overestimating the \texttt{Identity\_AA} and all other criteria is particularly
accentuated for the \texttt{macse\_p} method 
as compared to the \texttt{fse} and \texttt{fse0} methods, in all cases.
On the opposite, the  \texttt{needlenuc}
method always largely underestimates  the \texttt{Identity\_AA}, while overestimating all other criterion.

%\begin{landscape} 
%\end{landscape} 
\vspace{-0.3cm}
\paragraph{\bf Results for the \texttt{FS} datasets.} For the \texttt{FS} datasets, we assume
that the real alignments must contain FS translations. So, the \texttt{needleprot} method can no longer
produce the most accurate results. On the contrary, it is most likely that it
underestimates the \texttt{Identity\_AA} criterion. Indeed, it correctly aligns AA in CDS regions that are free of FS translation, but in FS translation regions, it either leads to several AA mismatches in the case of high mismatches scores, or to an overestimation of the \texttt{Gap\_init} criterion. 
As expected, we observed that the value of \texttt{Identity\_AA} for the 
\texttt{needleprot} method was always the lowest (data shown in the Additional file 3). We focus on the four other methods. 

Table \ref{tab:fsonly} summarizes 
the results for the nine cases considered.
For the \texttt{Identity\_NT} and \texttt{Identity\_AA} criteria,
the differences between the values for the four methods are negligible.
The main differences between the results reside in the values of the \texttt{Gap\_init} and \texttt{FS\_init} 
criteria. In particular, the \texttt{FS\_init} 
criterion is useful to compare the accuracy of the methods for correctly  identifying real FS translation regions. In \cite{okamura2006frequent} (for family I) and \cite{raes2005functional} (for families II to X), at most one FS translation region was detected and manually validated for each pair of CDS of the ten gene families.
 So, the expected number of FS 
translation regions per alignment in the \texttt{FS} data is
 $1$. In Table \ref{tab:fsonly}, we observe that, in all cases, the \texttt{fse} and  \texttt{fse0} methods are the only methods for which the average numbers of \texttt{FS\_init}
 are close to $1$ with +/- standard error values smaller than $1$. The \texttt{macse} method and especially the \texttt{needlenuc} method overestimate the number of FS translation regions per alignment with large standard error values in all cases.

\paragraph{\bf Results for the \texttt{ambiguFS} datasets.}
For the \texttt{ambiguFS} datasets, all methods do not agree
for the presence or absence of FS translation regions between the pairs of CDS. Note that the \texttt{needlenuc} method reports FS translations for all pairs of CDS, with the highest average number of FS translation regions per alignment in all cases (data shown in the Additional file 3). As \texttt{needlenuc} is already shown to perform poorly in both  the absence and the presence of FS translation regions, we focus on the four other methods. Table \ref{tab:fsambigu} summarizes the results. We observe that, for all criteria, \texttt{macse\_p} has higher values than \texttt{fse0},  \texttt{fse} and \texttt{needleprot} that have similar values. The most significant difference between the results resides in the values for the \texttt{FS\_init} and \texttt{FS\_length} criteria. The \texttt{fse} method always reports a null or a very small number of FS regions with an average \texttt{FS\_init} equals to $1$ as expected. In all cases, the \texttt{fse0} and \texttt{macse} methods overestimate the number of FS translation regions per alignment.

\begin{landscape} 
\begin{table}[ht!]
\flushleft
\caption{Values of the six criteria for the \texttt{\bf FS} dataset.}
\label{tab:fsonly}
\begin{tabular}{|l|l|l|l|l|l|l|l|l|}
\hline
\begin{tabular}[c]{@{}l@{}}fs\_open\_\\ cost\end{tabular} & \begin{tabular}[c]{@{}l@{}}fs\_extend\_cost\\ (\# CDS pairs)\end{tabular} & Method& \begin{tabular}[c]{@{}l@{}}Identity\_\\ NT\end{tabular} & \begin{tabular}[c]{@{}l@{}}Identity\_\\ AA\end{tabular} & \begin{tabular}[c]{@{}l@{}}Gap\_\\ init\end{tabular} & \begin{tabular}[c]{@{}l@{}}Gap\_\\ length\end{tabular} & \begin{tabular}[c]{@{}l@{}}FS\_init\\  (avg)\end{tabular}  & \begin{tabular}[c]{@{}l@{}}FS\_\\ length\end{tabular} \\ \hline \hline
\multirow{12}{*}{-10}    & \multirow{4}{*}{\begin{tabular}[c]{@{}l@{}}-1\\ (212)\end{tabular}}  & \texttt{fse0}       & 166002        & 325212        & 895        & 60662        & 226 (1.06 $\pm 0.25$)       & 20219       \\ \cline{3-9} 
       &    & \texttt{{\bf fse}}        & 165720        & 325026        & 901        & 60624        & 216 ({\bf 1.01} $\pm 0.14$)      & 18705       \\ \cline{3-9} 
       &    & \texttt{macse\_p}   & 166167        & 324999        & 1445       & 61562        & 432  (2.03 $\pm 3.06$ )     & 22742      \\ \cline{3-9} 
       &    & \texttt{needlenuc}  & 172959        & 321348        & 5053       & 60038        & 2103 (9.91 $\pm 26.73$)     & 29616       \\
%       &    & \textbf{MACSE  b}   & \textbf{163136}        & \textbf{318570}        & \textbf{1195}       & \textbf{64352}        & \textbf{211}       & \textbf{18501}       \\ 
\cline{2-9}  
       & \multirow{4}{*}{\begin{tabular}[c]{@{}l@{}}-0.5\\ (386)\end{tabular}}& \texttt{fse0}       & 252590        & 464712        & 2400       & 114859       & 482 (1.24 $\pm 0.47$)      & 31777       \\ \cline{3-9} 
       &    & \texttt{{\bf fse}}        & 251647        & 463407        & 2387       & 115269       & 401  ({\bf 1.03} $\pm 0.19$)     & 26982       \\ \cline{3-9} 
       &    & \texttt{macse\_p}   & 253715        & 465594        & 4161       & 117165       & 1306 (3.38 $\pm 4.53$)      & 41742       \\ \cline{3-9} 
       &    & \texttt{needlenuc}  & 279682        & 452673        & 19408      & 113195       & 8032  (20.80 $\pm 31.02$)  & 68226       \\ 
%       &    & \textbf{MACSE  b}   & \textbf{244588}        & \textbf{446775}        & \textbf{3084}       & \textbf{123279}       & \textbf{333}       & \textbf{22370}       \\ 
\cline{2-9}
       & \multirow{4}{*}{\begin{tabular}[c]{@{}l@{}}-0.2\\ (619)\end{tabular}}& \texttt{fse0}       & 371062        & 641748        & 5334       & 204370       & 805   (1.30 $\pm 0.52$)    & 43381       \\ \cline{3-9} 
       &    & \texttt{{\bf fse}}      & 370260        & 640377        & 5270       & 204806       & 688  (\textbf{1.11} $\pm 0.33$)     & 37376       \\ \cline{3-9} 
       &    & \texttt{macse\_p}   & 374729        & 646893        & 9308       & 208344       & 2893  (4.67 $\pm 5.34$)   & 72030       \\ \cline{3-9} 
       &    & \texttt{needlenuc}  & 442564        & 618270        & 48799      & 209420       & 19751  (31.90 $\pm 34.48$)   & 141217      \\ 
%       &    & \textbf{MACSE  b}   & \textbf{355338}        & \textbf{606516}        & \textbf{6432}       & \textbf{226800}       & \textbf{448}       & \textbf{27007}       \\ 
\hline
\hline
\multirow{12}{*}{-20}    & \multirow{4}{*}{\begin{tabular}[c]{@{}l@{}}-1\\ (161)\end{tabular}}  & \texttt{fse0}       & 123814        & 244350        & 461        & 40315        & 168   (1.04 $\pm 0.20$)     & 17770       \\ \cline{3-9} 
       &    & \texttt{{\bf fse}}  & 123610        & 244149        & 468        & 40195        & 164 (\textbf{1.01} $\pm 0.14$)      & 16924       \\ \cline{3-9} 
       &    & \texttt{macse\_p}   & 123541        & 243591        & 709        & 40585        & 223 (1.38 $\pm 1.03$)      & 18119       \\ \cline{3-9} 
       &    & \texttt{needlenuc}  & 125452        & 242742        & 1493       & 39031        & 650 (4.03 $\pm 5.85$)      & 19405       \\ 
%       &    & \textbf{MACSE  b}   & \textbf{122221}        & \textbf{240672}        & \textbf{663}        & \textbf{42603}        & \textbf{189}       & \textbf{17292}       \\ 
\cline{2-9}
       & \multirow{4}{*}{\begin{tabular}[c]{@{}l@{}}-0.5\\ (189)\end{tabular}}      & \texttt{fse0}       & 147476        & 291147        & 549        & 49485        & 197  (1.04 $\pm 0.20$)     & 19599       \\ \cline{3-9} 
       &    & \texttt{{\bf fse}}  & 147401        & 291048        & 557        & 49363        & 194    (\textbf{1.02} $\pm 0.16$)   & 19279       \\ \cline{3-9} 
       &    & \texttt{macse\_p}   & 147143        & 290271        & 838        & 49841        & 260  (1.37 $\pm 0.98$)    & 19976       \\ \cline{3-9} 
       &    & \texttt{needlenuc}  & 149551        & 289086        & 1872       & 47515        & 808  (4.27 $\pm 6.17$)     & 21440       \\
%       &    & \textbf{MACSE  b}   & \textbf{145577}        & \textbf{286824}        & \textbf{800}        & \textbf{52019}        & \textbf{212}       & \textbf{18798}       \\ 
\cline{2-9} 
       & \multirow{4}{*}{\begin{tabular}[c]{@{}l@{}}-0.2 \\ (216)\end{tabular}}        & \texttt{fse0}       & 161906        & 318117        & 723        & 55383        & 225 (1.04 $\pm 0.20$)      & 21300       \\ \cline{3-9} 
       &    & \texttt{{\bf fse}}  & 161865        & 318099        & 732        & 55393        & 223 (   \textbf{1.03} $\pm 0.18$)   & 21115       \\ \cline{3-9} 
       &    & \texttt{macse\_p}   & 161622        & 317205        & 1061       & 55715        & 306  (1.41 $\pm 0.99$)     & 21997       \\ \cline{3-9} 
       &    & \texttt{needlenuc}  & 165260        & 315531        & 2851       & 53613        & 1186 (5.49 $\pm 6.82$)     & 24403       \\  
         %\cline{3-9} 
%       &    & \textbf{MACSE  b}   & \textbf{159347}        & \textbf{312375}        & \textbf{992}        & \textbf{58049}        & \textbf{225}       & \textbf{19456}       \\ 
\hline
\hline
\multirow{12}{*}{-30}    & \multirow{4}{*}{\begin{tabular}[c]{@{}l@{}}-1\\ (71)\end{tabular}}   & \texttt{fse0}       & 47071& 91266& 230        & 26303        & 76     (1.07 $\pm 0.26$)   & 12845       \\ \cline{3-9} 
       &    & \texttt{{\bf fse}}  & 46872& 91032& 233        & 26183        & 72  (\textbf{1.01} $\pm 0.12$)      & 12302       \\ \cline{3-9} 
       &    & \texttt{macse\_p}   & 46936& 90876& 372        & 26325        & 118  (1.66 $\pm 1.25$)    & 13142       \\ \cline{3-9} 
       &    & \texttt{needlenuc}  & 48290& 91017& 866        & 26135        & 391   (5.50 $\pm 5.67$)    & 13829       \\
         %\cline{3-9} 
%       &    & \textbf{MACSE  b}   & \textbf{46263}& \textbf{89475}& \textbf{323}        & \textbf{28431}        & \textbf{118}       & \textbf{13444}       \\ 
\cline{2-9} 
       & \multirow{4}{*}{\begin{tabular}[c]{@{}l@{}}-0.5\\ (154)\end{tabular}}& \texttt{fse0}       & 120558        & 237768        & 445        & 37975        & 159    (1.03 $\pm 0.18$)   & 17554       \\ \cline{3-9} 
       &    & \texttt{{\bf fse}}  & 120504        & 237678        & 452        & 37851        & 157     (\textbf{1.01} $\pm 0.14$)  & 17319       \\ \cline{3-9} 
       &    & \texttt{macse\_p}   & 120338        & 237084        & 691        & 38047        & 212     (1.37 $\pm 1.00$)  & 17926       \\ \cline{3-9} 
       &    & \texttt{needlenuc}  & 122084        & 236904        & 1321       & 37531        & 575     (3.73 $\pm 5.14$)  & 18877       \\ %\cline{3-9} 
%       &    & \textbf{MACSE  b}   & \textbf{119180}        & \textbf{234630}        & \textbf{623}        & \textbf{40275}        & \textbf{189}       & \textbf{17473}       \\ 
\cline{2-9} 
       & \multirow{4}{*}{\begin{tabular}[c]{@{}l@{}}-0.2\\ (178)\end{tabular}}& \texttt{fse0}       & 137451        & 271041        & 525        & 46049        & 184 (1.03 $\pm 0.18$)       & 18995       \\ \cline{3-9} 
       &    & \texttt{{\bf fse}}  & 137440        & 271008        & 531        & 45917        & 183 (\textbf{1.02} $\pm 0.17$)      & 18872       \\ \cline{3-9} 
       &    & \texttt{macse\_p}   & 137175        & 270258        & 803        & 46187        & 244  ( 1.37 $\pm 0.97$)    & 19395       \\ \cline{3-9} 
       &    & \texttt{needlenuc}  & 139489        & 269139        & 1803       & 44303        & 779     (4.38 $\pm 6.27$)  & 20859       \\  
         %\cline{3-9} 
%       &    & \textbf{MACSE\_b}   & \textbf{135676}        & \textbf{266952}        & \textbf{753}        & \textbf{48559}        & \textbf{202}       & \textbf{18371}       \\ 
\hline
\end{tabular}
\\  For varying values of the parameters \texttt{fs\_open\_cost} and \texttt{fs\_extend\_cost}, the number of CDS pairs in the dataset is given. The values of the criteria for the \texttt{fse}, \texttt{fse0}, \texttt{macse\_p}, \texttt{needlenuc} methods are indicated. For each method, the average number of \texttt{FS\_init} per alignment, with corresponding standard error values are also indicated.
\end{table}
\end{landscape}

\begin{landscape} 
\begin{table}[ht!]
\flushleft
\caption{Values of the six criteria for the \texttt{\bf ambiguFS} dataset.}
\label{tab:fsambigu}

\begin{tabular}{|l|l|l|l|l|l|l|l|l|}
\hline
\begin{tabular}[c]{@{}l@{}}fs\_open\_\\ cost\end{tabular} & \begin{tabular}[c]{@{}l@{}}fs\_extend\_cost\\ (\# CDS pairs)\end{tabular} & Method& \begin{tabular}[c]{@{}l@{}}Identity\_\\ NT\end{tabular} & \begin{tabular}[c]{@{}l@{}}Identity\_\\ AA\end{tabular} & \begin{tabular}[c]{@{}l@{}}Gap\_\\ init\end{tabular} & \begin{tabular}[c]{@{}l@{}}Gap\_\\ length\end{tabular} & \begin{tabular}[c]{@{}l@{}}FS\_init\\  (avg)\end{tabular}  & \begin{tabular}[c]{@{}l@{}}FS\_\\ length\end{tabular} \\ \hline \hline
\multirow{12}{*}{-10}    & \multirow{4}{*} {-1 (2127)}& fse0 (862) & 1095102       & 1737105       & 24489      & 908218       & 1111   (1.28 $\pm 0.54$)       & 42730       \\ \cline{3-9} 
       &    & fse  & 1086546       & 1719774       & 23483      & 906540       & 0& 0  \\ \cline{3-9} 
       &    & macse\_p (2076)     & 1124316       & 1790199       & 45335      & 936002       & 12436   (5.99 $\pm 4.96$)  & 216772      \\ \cline{3-9} 
%       &    & needlenuc (2127)    & 1527999       & 1617300       & 273104     & 943464       & 110296 (51.85)    & 689967      \\ \cline{3-9} 
       &    & needleprot    & 1085007       & 1723950       & 25288      & 916518       & 0& 0     \\ \cline{2-9} 
       & \multirow{4}{*}{-0.5 (1953)}        & fse0 (688)   & 1008514       & 1597605       & 22984      & 854021       & 855   (1.24 $\pm 0.53$ )    & 31172       \\ \cline{3-9} 
       &    & fse (2)      & 1003293       & 1587258       & 22102      & 853793       & 2     (1.0 $\pm 0$ )    & 80 \\ \cline{3-9} 
       &    & macse\_p  (1902)       & 1036768       & 1649604       & 42619      & 880399       & 11562  (6.07 $\pm 4.91$)   & 197772      \\ \cline{3-9} 
%       &    & needlenuc (1953)    & 1421276       & 1485975       & 258749     & 890307       & 104367 (53.43)    & 651357      \\ \cline{3-9} 
       &    & needleprot    & 1001957       & 1591134       & 23790      & 863199       & 0& 0         \\ \cline{2-9} 
       & \multirow{4}{*}{-0.2 (1720)}        & fse0  (455)   & 890042        & 1420569       & 20050      & 764510       & 532    (1.16 $\pm 0.48$)   & 19568       \\ \cline{3-9} 
       &    & fse (3)       & 887372        & 1415403       & 19465      & 764162       & 3  (1.0 $\pm 0$)       & 92 \\ \cline{3-9} 
       &    & macse\_p  (1669)       & 915754        & 1468305       & 37472      & 789220       & 9975  (5.97 $\pm 4.75$)    & 167484      \\ \cline{3-9} 
%       &    & needlenuc (1720)    & 1258394       & 1320378       & 229358     & 794082       & 92648 (53.86)    & 578366      \\ \cline{3-9} 
       &    & needleprot    & 886178        & 1418748       & 20955      & 772272       & 0& 0         \\ \hline \hline
\multirow{12}{*}{-20}    & \multirow{4}{*}{-1 (409)}  & fse0 (100)    & 219277        & 358554        & 3633       & 153487       & 120     (1.2 $\pm 0.40$)  & 6937        \\ \cline{3-9} 
       &    & fse  & 216936        & 353586        & 3619       & 152391       & 0& 0  \\ \cline{3-9} 
       &    & macse\_p (403) & 225976        & 374391        & 6509       & 158165       & 1348   (3.34 $\pm 3.00$)   & 36179       \\ \cline{3-9} 
%       &    & needlenuc     & 289271        & 344037        & 42656      & 150413       & 18033 (44.09)    & 119142      \\ \cline{3-9} 
       &    & needleprot    & 216842        & 355656        & 4172       & 153957       & 0& 0         \\ \cline{2-9} 
       & \multirow{4}{*}{-0.5 (381)}& fse0 (72)     & 195615        & 311757        & 3545       & 144317       & 91 (1.26 $\pm 0.44$)        & 5108        \\ \cline{3-9} 
       &    & fse  & 194048        & 308448        & 3505       & 144045       & 0& 0  \\ \cline{3-9} 
       &    & macse\_p (375)  & 202374        & 327711        & 6380       & 148909       & 1311 (3.49 $\pm 3.05$)     & 34322       \\ \cline{3-9} 
%       &    & needlenuc     & 265172        & 297693        & 42277      & 141929       & 17875  (46.91)   & 117107      \\ \cline{3-9} 
       &    & needleprot    & 193980        & 310632        & 4051       & 145563       & 0& 0        \\ \cline{2-9} 
       & \multirow{4}{*}{-0.2 (354)}& fse0 (45)   & 181185        & 284787        & 3371       & 138419       & 63   (1.4 $\pm 0.49$)     & 3407        \\ \cline{3-9} 
       &    & fse (1)       & 180151        & 282693        & 3344       & 138217       & 1    (1.0 $\pm 0$)    & 40 \\ \cline{3-9} 
       &    & macse\_p (348) & 187895        & 300777        & 6157       & 143035       & 1265  (3.63 $\pm 3.11$)  & 32301       \\ \cline{3-9} 
%       &    & needlenuc     & 249463        & 271248        & 41298      & 135831       & 17497 (49.42)    & 114144      \\ \cline{3-9} 
       &    & needleprot    & 180116        & 284946        & 3883       & 139731       & 0& 0         \\ \hline \hline
\multirow{15}{*}{-30}    & \multirow{4}{*}{-1 (200)}  & fse0  (119)     & 151090        & 289617        & 805        & 42437        & 120    (1.01 $\pm 0.09$)   & 6818        \\ \cline{3-9} 
       &    & fse  & 147590        & 282018        & 852        & 40221        & 0& 0  \\ \cline{3-9} 
       &    & macse\_p  (200) & 152626        & 292254        & 1309       & 43043        & 378 (1.89 $\pm 2.16$)      & 14515       \\ \cline{3-9} 
%       &    & needlenuc     & 159929        & 289836        & 5288       & 42337        & 2241   (11.20) & 20701       \\ \cline{3-9} 
       &    & needleprot    & 147228        & 281472        & 933        & 40455        & 0& 0         \\ \cline{2-9} 
       & \multirow{4}{*}{-0.5 (117)}& fse0 (36)     & 77603& 143115        & 590        & 30765        & 37       (1.02 $\pm 0.16$) & 2109        \\ \cline{3-9} 
       &    & fse  & 76678& 141108        & 626        & 29913        & 0& 0  \\ \cline{3-9} 
       &    & macse\_p (117) & 79224& 146046        & 990        & 31321        & 284    (2.42 $\pm 2.65$)   & 9731        \\ \cline{3-9} 
%       &    & needlenuc     & 86135& 143949        & 4833       & 30941        & 2057 (17.58)     & 15653       \\ \cline{3-9} 
       &    & needleprot    & 76561& 141036        & 703        & 30099        & 0& 0     \\ \cline{2-9} 
       & \multirow{4}{*}{-0.2 (93)} & fse0 (12)   & 60710& 109842        & 510        & 22691        & 12    (1.0 $\pm 0$)    & 668\\ \cline{3-9} 
       &    & fse  & 60407& 109170        & 518        & 22491        & 0& 0  \\ \cline{3-9} 
       &    & macse\_p (93) & 62387& 112872        & 878        & 23181        & 252   (2.70 $\pm 2.89$)    & 8262        \\ \cline{3-9} 
%       &    & needlenuc     & 68730& 111714        & 4351       & 24169        & 1853 (19.92)     & 13671       \\ \cline{3-9} 
       &    & needleprot    & 60270& 109122        & 581        & 22677        & 0& 0        \\ \hline
\end{tabular}
\\ For varying values of the parameters \texttt{fs\_open\_cost} and \texttt{fs\_extend\_cost}, and for each method, the number of CDS pairs displaying a FS translation is given. The values of the criteria for each method are indicated. For each method, the average number of \texttt{FS\_init} per alignment, with corresponding standard error values are also indicated.
\end{table}
\end{landscape} 

\subsubsection*{\bf Third strategy: Using a 3-CDS manually-built benchmark}

We manually built the real pairwise alignments of three CDS from
three paralogous human genes of gene family I, the CDS  \emph{FAM86C1-002}
coding for protein \emph{ENSP00000352182.4}, \emph{FAM86B1-001} coding for
protein \emph{ENSP00000431362.1} and \emph{FAM86B2-202} coding for protein
\emph{ENSP00000311330.6}. The real multiple alignment of the three CDS is
roughly depicted and detailed in Figure \ref{fig:fam86-1}. 
From Figure
\ref{fig:fam86-1}, we observe 
that \emph{FAM86C1-002} shares
 with \emph{FAM86B1-001} a nucleotide region of length $159$ translated
 in the same frame and a nucleotide region of length $89$ with FS translation,
 while it only shares with \emph{FAM86B2-202} a nucleotide region of
 length $238$ ($89+149$) entirely under FS translation. It is then clear
 that CDS \emph{FAM86C1-002} and \emph{FAM86B1-001} are the most similar. Figure \ref{fig:fam86-1} also shows that each pair of CDS shares a single FS translation region.
  
 Table \ref{tab:result-fam86} shows the normalized pairwise similarity scores 
and the number of FS translation regions computed by the five alignment methods (the pairwise alignments computed by the five methods with varying \texttt{fs\_open\_cost}  and \texttt{fs\_extend\_cost} are given in the Additional file 4).
It shows that  \texttt{needleprot} and \texttt{fse} (in all cases where \texttt{fs\_extend\_cost}= -1) are 
the only two methods that allow to infer that  \emph{FAM86C1-002} and \emph{FAM86B1-001} are the most
 similar. Table \ref{tab:result-fam86} also illustrates the fact that \texttt{needlenuc} and \texttt{macse\_p} strongly overestimate the number of FS translation regions per alignment in all cases. The \texttt{fse} method with the parameters \texttt{fs\_open\_cost}= -10 and \texttt{fs\_extend\_cost}= -1
 is the only method that allows to infer that \emph{FAM86C1-002} and \emph{FAM86B1-001} are the most similar and to detect a single FS translation region for each alignment.

\begin{table}[ht!]
\flushleft
	\caption{Pairwise similarity scores and number of FS translation regions computed by the methods.}
	\begin{tabular}{|l|l|l|l|l|}
		\hline
	\texttt{fs\_open\_cost}	&\texttt{Method} & \emph{C1-002} vs \emph{B1-001} &  \emph{C1-002} vs \emph{B2-202} &  \emph{B1-001} vs \emph{B2-202}\\
		\hline
		\multirow{5}{*}{-10} &\texttt{fse0 } &  0.42 (1)& 0.58  (2)&  0.45 (1) \\
       	\cline{2-5}
		      &  \texttt{fse (-1)} & 0.33 (1) &0.27 (1) &0.18 (1)\\
      	\cline{2-5}
        		
		     &   \texttt{fse (-0.5)}& 0.37 (1) &0.43 (1) &0.31 (1)\\
                                
         \cline{2-5}
        		
		      &  \texttt{fse (-0.2)} &0.40 (1)  &0.52 (1)&0.39 (1) \\       

         \cline{2-5}
		      & \texttt{macse\_p} &0.40 (4) & 0.54(6) & 0.44 (1)	\\
     	\hline  \hline
        
       \multirow{5}{*}{-20} &\texttt{fse0 } & 0.39 (1) & 0.54  (1)& 0.41 (1)\\
       	\cline{2-5}
		      &  \texttt{fse (-1)} & 0.36 (0) &0.24 (1)  & 0.14 (1)\\
      	\cline{2-5}
        		
		     &   \texttt{fse (-0.5)} & 0.34 (1) &0.39  (1)&0.28 (1)\\
                                
         \cline{2-5}
        		
		      &  \texttt{fse (-0.2)} &  0.37 (1)& 0.48 (1) &0.36 (1)\\       

         \cline{2-5}
		      & \texttt{macse\_p} & 0.33 (4)& 0.47 (6)&0.35	(1)\\
     	\hline  \hline
       \multirow{5}{*}{-30} &\texttt{fse0 } & 0.35 (1) &0.50 (1)&  0.38 (1) \\
       	\cline{2-5}
		      &  \texttt{fse (-1)} & 0.36  (0)&0.20 (1) &0.11 (1)\\
      	\cline{2-5}
        		
		     &   \texttt{fse (-0.5)}& 0.36  (0)&  0.35 (1)& 0.25 (1)\\
                                
         \cline{2-5}
        		
		      &  \texttt{fse (-0.2)} & 0.33 (1) &0.44 (1) &0.33 (1) \\       

         \cline{2-5}
		      & \texttt{macse\_p} & 0.27 (4)& 0.39 (6)&	0.29 (1)\\
     	\hline  \hline
     &  \texttt{needlenuc} & 0.16  (23) &0.35  (15) & -0.36  (1) \\
		      
       	\hline \hline
		&\texttt{needleprot} & 0.38 (0)&-0.12  (0)  &-0.13  (0) \\
        \hline 
	\end{tabular}
        \\
        Normalized pairwise similarity scores and number of FS translation regions computed by the five methods
        for the 3-CDS  manually-built benchmark composed of CDS
        \emph{FAM86C1-002}, \emph{FAM86B1-001} and \emph{FAM86B2-202} (Similarity scores are normalized by dividing them by the lengths of alignments).
        \label{tab:result-fam86}
\end{table}

 \subsubsection*{\bf Fourth strategy: Inferring CDS splicing orthology groups and protein phylogenies}
 
Based on the three CDS used in the previous strategy, CDS 
\emph{FAM86C1-002}
from human gene \emph{ENSG00000158483}, \emph{FAM86B1-001} from human 
gene \emph{ENSG00000186523} and \emph{FAM86B2-202} from human gene 
\emph{ENSG00000145002}, we generated a dataset of three CDS splicing 
orthology groups composed of 
$21$ homologous CDS. Each group contains one of the three initial CDS 
and its  six splicing orthologs in the following set of seven genes from 
gene family I : human genes \emph{ENSG00000158483} denoted $H1$, 
\emph{ENSG00000186523} denoted $H2$ and  \emph{ENSG00000145002} denoted 
$H3$, each containing one of the initial CDS, chimpanzee gene 
\emph{ENSPTRG00000007738} denoted $Ch$, mouse gene 
\emph{ENSMUSG00000022544} denoted $M$, rat gene 
\emph{ENSRNOG0-0000002876} denoted $R$ and cow gene 
\emph{ENSBTAG00000008222} denoted $Co$. The CDS splicing orthologs were 
predicted based on the spliced alignment tool Splign \cite{kapustin2008} 
as follows: for each initial CDS $A_1$ of a gene $A$ and each gene $B$ 
different from $A$, $A_1$ was aligned to $B$ and a putative or existing 
CDS of $B$ ortholog to $A_1$ with the same splicing structure was inferred.
The $21$ resulting CDS are given in Additional file 5. 

We computed the normalized pairwise similarity scores between 
the CDS, using the five alignment methods (the pairwise alignments computed by the five methods with varying \texttt{fs\_open\_cost}  and \texttt{fs\_extend\_cost} are given in the Additional file 5). For each method, we constructed a phylogeny using an UPGMA and a Neighbor-Joining (NJ) algorithm, %\cite{saitou1987}
based on the computed CDS similarity matrix. The UPGMA algorithm
was used to classify the CDS into three groups and infer the
similarity relationships between the groups independently of any 
rate of evolution. The NJ algorithm was used to reconstruct the 
phylogeny inside each group. Table \ref{tab:result-phylo-fam86} 
summarizes the results.  The three splicing orthology groups are
denoted \emph{G1} (containing CDS \emph{C1-002}), \emph{G2} 
(containing CDS \emph{B1-001}) and \emph{G3} (containing CDS \emph{B2-202}).

All methods allow to correctly classify the CDS into the three
initial splicing orthology groups G1, G2, and G3. However, the 
\texttt{needleprot} and \texttt{fse} methods are the only methods that 
allow to infer the correct similarity relationships ((G1,G2),G3)
between the groups, confirming the results of the third evaluation 
strategy. For all methods, the CDS phylogeny reconstructed inside the 
group G2 is (Co,((M,R),((H1,Ch),(H2,H3)))) inducing
an evolution of the seven genes with a speciation event at the root
of the gene tree. The phylogeny reconstructed for the groups G1 and G3 is 
((M,R),(Co,((H1,Ch),(H2,H3)))), inducing an evolution of the 
genes with a duplication event at the root of the phylogeny.

\begin{table}[ht!]
\flushleft
	\caption{Similarity relationships between the groups G1, G2 and G3 for the five methods.}
	\begin{tabular}{|l|l|l|l|}
		\hline
	\texttt{fs\_open\_cost}	&\texttt{Method}  &  ((G1,G3),G2) & ((G1,G2),G3)    \\
		\hline \hline
		 \multirow{5}{*}{-10} &  \texttt{fse (-1)}&    &X\\
      		\cline{2-4}
        		
		     &   \texttt{fse (-0.5)}&    &X \\
                                
        	\cline{2-4}
        		
		      &  \texttt{fse (-0.2)} & X  & \\       
       	\cline{2-4}
		      &\texttt{fse0 } &X & \\

         \cline{2-4}
		      & \texttt{macse\_p} & X &  \\
     	\hline  \hline
        
		\multirow{5}{*}{-20} &  \texttt{fse (-1)} &    &X \\
      	\cline{2-4}
        		
		     &   \texttt{fse (-0.5)}&  &X  \\
                                
         \cline{2-4}
        		
		      &  \texttt{fse (-0.2)} &    & X\\       

       	\cline{2-4}
		      &\texttt{fse0 } & X   &  \\
         \cline{2-4}
		      & \texttt{macse\_p} & X  &  \\
     	\hline  \hline

		\multirow{5}{*}{-30} &  \texttt{fse (-1)} &   &X  \\
      	\cline{2-4}
        		
		     &   \texttt{fse (-0.5)}&  &  X\\
                                
         \cline{2-4}
        		
		      &  \texttt{fse (-0.2)} &   & X \\       

       	\cline{2-4}
		      &\texttt{fse0 } &     & X \\
         \cline{2-4}
		      & \texttt{macse\_p} & X  & \\
     	\hline  \hline

     &  \texttt{needlenuc} &   X  &  \\
		      
       	\hline \hline
		&\texttt{needleprot} &  &  X \\
        \hline 
	\end{tabular}
    \\
        Similarity relationships between the splicing orthology groups 
        G1, G2 and G3 computed using the similarity matrices of the five methods for 
        the 21-CDS dataset.
        \label{tab:result-phylo-fam86}
\end{table}

 \subsubsection*{\bf Comparing of the running times}

Table \ref{tab:timeexecution} shows the running times for each of the five methods on the three first gene families of our dataset on a $24 \times 2.1$GHz processor with $10$GB of RAM.  The \texttt{needleprot} method is the fastest, followed by \texttt{macse\_p} and then \texttt{needlenuc}, while \texttt{fse} and \texttt{fse0}  are  the slowest methods. 

 Note that for \texttt{fse}, \texttt{fse0}, \texttt{needlenuc} and \texttt{needleprot}, the used implementations are in Python, while we used a JAVA implementation for \texttt{macse\_p} provided by its authors. This explains the fact that \texttt{macse\_p} is unexpectedly faster here than \texttt{fse}, \texttt{fse0}, and even  \texttt{needlenuc}. Indeed, the five methods share the same asymptotic time complexity, but the exact complexity of each of them is dependent on the number of calls of the main recurrence formulas in an execution, and the number of cases considered in each recurrence formula. The exact computational complexity of the five methods in terms of the lengths $n$ and $m$ of two compared CDS are $12.55\times nm$ for \texttt{fse} and \texttt{fse0} (as shown in the proof of Theorem \ref{thm}), $15\times nm$ for \texttt{macse\_p}, $3\times nm$ for \texttt{needlenuc} and $0.33\times nm$ for \texttt{needleprot}.

\begin{table}[ht!]
\flushleft
\caption{Running time in seconds for each method.}
\begin{tabular}{|l|l|l|l|l|l|}
\hline
   Gene family      & \texttt{fse0}  & \texttt{fse}   &\texttt{macse\_p} & needlenuc & needleprot \\ \hline
$I$ & 299  & 291 & 53 & 97    & 22     \\ \hline
$II$ & 270 & 260 & 45 & 93    & 20     \\ \hline
$III$ & 377 & 389 & 54 & 62    & 20     \\ \hline
\end{tabular}
\\
For each method and gene families I, II, and III, the running time was calculated on the same computer (24 processors of $2.1$GHz each and $10$GB of RAM) with the parameters \texttt{fs\_open\_cost}=$ -20$ and \texttt{fs\_extend\_cost}=$-0.2$. 
\label{tab:timeexecution}
\end{table}

\section*{Conclusions}
\label{conclusion}

In this paper, we introduce a new scoring model for the alignment
of CDS accounting for frameshift translation length. The motivation
for this new scoring scheme is the increasing evidence for protein
divergence through frameshift translation in eukaryotic coding gene
families, calling for automatic methods able to compare, align and
classify CDS while accounting for their codon structure.
The aim of this paper is to validate the necessity of accounting for
frameshift translation length when comparing CDS and show that computing a
maximum score pairwise alignment under the new scoring scheme is possible
in quadratic time complexity.
The results of comparing five
CDS alignment methods for the pairwise alignment of CDS from
ten eukaryotic gene families show that our method is the best compromise
for sets of CDS in which some pairs of CDS display FS translations
while some do not. Future work will make use of benchmarks of CDS
alignments generated manually and by simulation in order to confirm
these experimental results. We also defer to a future work the
extended study of our model's robustness to parameter changes and the 
calibration of its parameters using real data benchmarks.
The perspectives of this work also include the design of a heuristic 
algorithm using local alignment that will achieve scalability for 
large datasets while keeping high accuracy, and
the extension of the method toward multiple alignment. Finally, we plan to apply the algorithms for the discovery of non-annotated frameshifts, and the evaluation of the extent of frameshifts in eukaryotic gene families.

%%%%%%%%%%%%%%%%%%%%%%%%%%%%%%%%%%%%%%%%%%%%%%
%%%%
%% Backmatter begins here %%
%%%%
%%%%%%%%%%%%%%%%%%%%%%%%%%%%%%%%%%%%%%%%%%%%%%

\begin{backmatter}

\section*{Availability of supporting data}
An implementation of the pairwise alignment method  in Python is available at \url{https://github.com/UdeS-CoBIUS/FsePSA}. The dataset used in section Results is available in the Additional files.

\section*{List of abbreviations}

CDS: Coding DNA Sequence;
FS: Frameshift;
NT: nucleotide;
AA: amino acid;
NW: Needleman-Wunsch.

\section*{Declarations} 
\hfill
\section*{Author's contributions}
SJ, EK, FB and AO wrote the program and its documentation. SJ and AO conceived the study and its design. SJ, EK and AR ran the experiments. SJ, EK and AO analyzed and interpreted the data. SJ and AO wrote the manuscript. SJ, EK, MS and AO  critically revised the manuscript. All authors read and approved the final manuscript.
\section*{Acknowledgements}
EK has a scholarship from the Faculty of Science of Université de Sherbrooke. AO is funded by the Canada Research Chair in Computational and Biological Complexity and the Université de Sherbrooke.
\section*{Competing interests}
The authors declare that they have no competing interests.
\section*{Ethics approval and consent to participate}
Not applicable.
\section*{Consent for publication}
Not applicable.
\section*{Funding}
Research funded by the Canada Research Chairs (CRC) (CRC Tier2 Grant 950-230577) and Université de Sherbrooke..

%%%%%%%%%%%%%%%%%%%%%%%%%%%%%%%%%%%%%%%%%%%%%%%%%%%%%%%%%%%%%
%%The Bibliography     %%
%%      %%
%%  Bmc_mathpys.bst  will be used to     %%
%%  create a .BBL file for submission.   %%
%%  After submission of the .TEX file,   %%
%%  you will be prompted to submit your .BBL file.%%
%%      %%
%%      %%
%%  Note that the displayed Bibliography will not %%
%%  in the online Instructions for Authors.       %%
%%      %%
%%%%%%%%%%%%%%%%%%%%%%%%%%%%%%%%%%%%%%%%%%%%%%%%%%%%%%%%%%%%%

% if your bibliography is in bibtex format, use those commands:
\bibliographystyle{bmc-mathphys} % Style BST file (bmc-mathphys, vancouver, spbasic).
\bibliography{biblio}      % Bibliography file (usually '*.bib' )
% for author-year bibliography (bmc-mathphys or spbasic)
% a) write to bib file (bmc-mathphys only)
% @settings{label, options="nameyear"}

%%%%%%%%%%%%%%%%%%%%%%%%%%%%%%%%%%%
%%    %%
%% Figures     %%
%%    %%
%% NB: this is for captions and  %%
%% Titles. All graphics must be  %%
%% submitted separately and NOT  %%
%% included in the Tex document  %%
%%    %%
%%%%%%%%%%%%%%%%%%%%%%%%%%%%%%%%%%%

%%
%% Do not use \listoffigures as most will included as separate files

%\section*{Figures}

%%%%%%%%%%%%%%%%%%%%%%%%%%%%%%%%%%%
%%    %%
%% Tables      %%
%%    %%
%%%%%%%%%%%%%%%%%%%%%%%%%%%%%%%%%%%

%\section*{Tables}

%
%%%%%%%%%%%%%%%%%%%%%%%%%%%%%%%%%%%
%%    %%
%% Addition-al Files     %%
%%    %%
%%%%%%%%%%%%%%%%%%%%%%%%%%%%%%%%%%%

\section*{Additional Files}

\subsection*{Additional file 1 -- Proof of Lemma \ref{D}}
PDF file containg the detailed proof of Lemma \ref{D}.

\subsection*{Additional file 2 -- CDS of the ten gene families}
Zip file containing the CDS files at the fasta format for each of the ten gene families considered in the Results section.

\subsection*{Additional file 3 -- Additional lines for Tables \ref{tab:fsonly}  and \ref{tab:fsambigu}}
PDF file containing additional lines for Tables \ref{tab:fsonly} (for \texttt{needleprot}) and \ref{tab:fsambigu}  (for \texttt{needlenuc}) of the Results section.

\subsection*{Additional file 4 -- Pairwise alignments for the 3-CDS benchmark}
Zip file containing the sequence file and the pairwise alignment files at the fasta format for the manually-built 3-CDS benchmark considered in the Results section, for each of the five methods and each parameter configuration.

\subsection*{Additional file 5 -- Pairwise alignments for the 21-CDS dataset}
Zip file containing the sequence file and the pairwise alignment files at the fasta format for the 21-CDS benchmark considered in the Results section, for each of the five methods and each parameter configuration.

\end{backmatter}
\end{document}